\documentclass{article}
\usepackage{kr}

\usepackage{times}
\usepackage{soul}
\usepackage{url}
\usepackage[hidelinks]{hyperref}
\usepackage[utf8]{inputenc}
\usepackage[small]{caption}
\usepackage{graphicx}
\usepackage{amsmath}
\usepackage{amsthm,amssymb}
\usepackage{thmtools,thm-restate}
\usepackage{stmaryrd}
\usepackage{booktabs}
\usepackage{algorithm}
\usepackage{algorithmic}
\usepackage{color}
\newcommand{\new}[1]{{#1}}

\usepackage{latexsym}

\usepackage{xspace}

\newcommand{\comment}[1]{}

\newcommand{\Aa}{{\cal{A}}}
\newcommand{\Bb}{{\cal{B}}}

\newcommand{\Ee}{{\cal{E}}}
\newcommand{\Ff}{{\cal{F}}}
\newcommand{\Ii}{{\cal{I}}}
\newcommand{\Jj}{{\cal{J}}}
\newcommand{\Kk}{{\cal{K}}}

\newcommand{\Mm}{{\cal{M}}}

\newcommand{\Tt}{{\cal{T}}}

\newcommand{\pp}{\mathbf{p}}

\newcommand{\fentails}{\models_{\mathsf{fin}}}

\newcommand{\poly}{\mathrm{poly}}

\newcommand{\ALC}{\ensuremath{{\cal{ALC}}\xspace}}

\newcommand{\mn}[1]{\mathsf{#1}}

\newcommand{\concepts}{\ensuremath{\mn{N_{\mn{C}}}}}
\newcommand{\conceptsBb}{\ensuremath{\CN(\widehat\Bb)}}

\newcommand{\roles}{\ensuremath{\mn{N_{\mn{R}}}}}

\newcommand{\var}{\textit{var}}

\newcommand{\Ind}{\mn{ind}}
\newcommand{\Rol}{\mn{rol}}

\newcommand{\CN}{\mn{CN}}

\newcommand{\tp}{\mn{tp}}

\definecolor{armygreen}{rgb}{0.29, 0.33, 0.13}

\newtheorem{theorem}{Theorem}
\declaretheorem{fact}

\title{Finite Entailment of UCRPQs over $\ALC$  Ontologies}

\author{%
V{\'i}ctor Guti{\'e}rrez-Basulto$^1$\and
Albert Gutowski$^2$\and
Yazm{\'i}n Ib{\'a}{\~n}ez-Garc{\'i}a$^1$\and
Filip Murlak$^2$ \\
\affiliations
$^1$Cardiff University, UK\\
$^2$University of Warsaw, Poland\\
\emails
  \emails
\{a.gutowski, f.murlak $\!\!$\}@mimuw.edu.pl, \{gutierrezbasultov, ibanezgarciay $\!\!$\}@cardiff.ac.uk
}

\begin{document}

\maketitle

\begin{abstract}
We investigate the problem of \emph{finite} entailment of ontology-mediated queries. We consider the expressive query language, unions of conjunctive regular path queries (UCRPQs), extending  the well-known class of  union of conjunctive queries, with regular expressions over roles. We look at ontologies formulated using the description logic \ALC, and show a tight  \textsc{2ExpTime} upper bound for entailment of UCRPQs. At the core of  our decision procedure, there  is  a novel automata-based technique introducing a stratification of interpretations induced by the deterministic finite automaton underlying the input UCRPQ. 
\end{abstract}

\section{Introduction}
\label{sec:introduction}

At the intersection of knowledge representation  and database theory lies the fundamental problem  of \emph{ontology-mediated query entailment (OMQE)}, where  the background knowledge provided by an ontology  is used to enrich the answers to queries posed to databases. 
In this context, description logics (DLs) are a widely accepted family of logics used to formulate ontologies.  By now, the OMQE problem under the unrestricted semantics (reasoning over arbitrary models)  is well understood for various query languages and DLs~\cite{DBLP:journals/ki/SchneiderS20a}. In contrast, for the \emph{finite} OMQE problem, where one is interested in reasoning over finite models only, the overall landscape is rather incomplete. However, in recent years, the study of finite OMQE has been gaining traction, considering both lightweight and expressive DLs and (mostly) unions of conjunctive queries~\cite{DBLP:conf/esws/Rosati08,DBLP:conf/kr/GarciaLS14,DBLP:conf/kr/Rudolph16,GogaczIM18,DBLP:conf/ijcai/GogaczGIJM19,DBLP:conf/mfcs/DanielskiK19,GogaczGGIM20,DBLP:conf/dlog/BednarczykK21}.

In this paper we consider the problem of  finite OMQE with unions of conjunctive regular path queries (UCRPQs) as the query language. 
UCRPQs~\cite{DBLP:conf/pods/FlorescuLS98,DBLP:conf/kr/CalvaneseGLV00} are a powerful navigational query language  for  graph databases in which one can express that two entities are related by a path of edges that can be specified by a regular language over  binary relations. So, UCRPQs extend unions of conjunctive queries (UCQs) with atoms that might contain regular expressions that traverse the edges of the database. Indeed, path navigation is included  in the  query language XPath 2.0 for XML data, and it is also present in the  SPARQL 1.1 query language for RDF data through the property path feature.  Given the resemblance of instance data stored in ABoxes in DLs  to graph-like data, several investigations on unrestricted entailment of various types of navigational query languages mediated by DL ontologies have been carried out~\cite{DBLP:journals/jair/StefanoniMKR14,DBLP:journals/iandc/CalvaneseEO14,DBLP:journals/jair/BienvenuOS15,DBLP:conf/aaai/Gutierrez-Basulto18,DBLP:conf/ijcai/GogaczGIJM19,DBLP:conf/ijcai/BednarczykR19}, yielding algorithmic approaches and optimal complexity bounds. For finite entailment of regular path queries mediated by DL ontologies, there are only undecidability  results available~\cite{DBLP:conf/kr/Rudolph16}. The most relevant positive news are the decidability and  computational complexity results by \citeauthor{DBLP:conf/mfcs/DanielskiK19}~(\citeyear{DBLP:conf/mfcs/DanielskiK19}) and~\citeauthor{GogaczGGIM20}~(\citeyear{GogaczGGIM20}) on finite entailment of conjuctive queries with transitive closure over roles mediated by expressive DL ontologies.

We focus on ontologies formulated using the description logic \ALC{}.  Note that entailment of UCRPQs over \ALC{} ontologies is not \emph{finitely controllable}, i.e.\ finite and unrestricted entailment do not coincide as it is \emph{not} the case that  for any \ALC{} knowledge base $\Kk$ and any UCRPQ $\varphi$, it
holds that $\Kk$ entails  $\varphi$ over all (unrestricted) models iff $\Kk$ entails $\varphi$ over all finite models. By assuming that the represented world is finite, we can therefore not reuse existing complexity bounds or  algorithmic approaches to UCRPQ entailment. From a usability perspective, the suitability of this assumption depends on the potential applications. A particular interest for navigational queries comes from bioinformatics and cheminformatics
  \cite{Lysenko2016,Bio1,bio2,DBLP:conf/semweb/HuQD15,doi:10.1177/0165551519865495,bio3}. For instance, experts often need to find associations between entities in protein, cellular, drug, and disease networks (represented as graph databases), so that e.g.\  gene-disease-drug associations (corresponding to paths in the database) can be discovered for  developing new treatment methods. In this  type of applications, databases and the models they represent are clearly meant to be finite. Importantly, biochemical networks contain complex motifs involving e.g.\ \emph{cycles} or cliques. This type of  patterns can be described using UCRPQs, however,
   without the finiteness assumption   these patterns could be disregarded as the associated query might not be entailed when reasoning over all  models (including infinite ones).


\subsection*{Contribution} 

The  main technical contribution of our investigation is the development of a dedicated   automata-based method for entailment of UCRPQs over \ALC{} ontologies, providing an optimal upper bound. More precisely, we obtain  the following result, where the matching lower bound is inherited from~\cite{DBLP:conf/rr/OrtizS14}. 
\begin{theorem}~\label{thm:mainresult}
Finite entailment of UCRPQs over \ALC{} ontologies is \textsc{2ExpTime}-complete. 
\end{theorem}
In prior work, \citeauthor{DBLP:conf/kr/Rudolph16}~(\citeyear{DBLP:conf/kr/Rudolph16}) showed that finite entailment of 2RPQs in $\mathcal{ALCIO\hspace{-1pt}F}$   is undecidable.
Theorem~\ref{thm:mainresult}  thus provides a key step towards delimiting the decidability boundary of finite OMQE with navigational queries.

\smallskip

At the heart of our approach to finite entailment of UCRPQs in \ALC{} there  is  a stratification of interpretations induced by the deterministic finite automaton underlying the  UCRPQ.
This stratification  builds upon the so-called \emph{tape construction}, previously used 
to efficiently evaluate  queries in the extension of XPath 1.0 where arbitrary regular expressions may appear as path expressions~\cite{DBLP:journals/jacm/BojanczykP11}.
To realize the tape construction,  our method represents UCRPQs by means of a  semiautomaton $\Bb$~\cite{AlgebraicAutomata} and defines an expansion of $\Bb$, allowing to trace runs of $\Bb$ that begin in all possible states, on all infixes of the input word. We make interpretations $\Ii$  knowledgeable of  the expansion by enriching paths of $\Ii$ with  its possible runs and by  associating   edges of $\Ii$ with levels $\ell$ induced by the transitions of the expansion. 
In a similar fashion we also make CRPQs sensible of levels. With this at hand,  we tackle finite entailment by eliminating the lowest level from a query and from an interpretation, and then recursively solving the simpler problem. At each step of this process, we should be able to arrange solutions to simpler problems in a hierarchical way so that we can reason over them. To this aim, we consider a variant of entailment that includes an \emph{environment}, which will provide the necessary information to position the arranged solutions to simpler problems in the context of larger  interpretations. To better keep track of the complexity of our recursive method, we introduce a modification of the entailment problem modulo environment in which we look at a  particular type of finite models: \emph{$(\ell,\ell')$-models}, which are models  with edges of levels $\ell$ or higher that are `consistent' w.r.t.\  queries referring to edges of level $\ell'$ or higher. We solve the problem of finding $(\ell,\ell')$-models  recursively by increasing $\ell$ and $\ell'$ in an alternating way, until both reach the maximum level $n+1$, with $n$ the number of states of $\Bb$. This will mean solving  finite entailment modulo environment, and thus standard finite entailment as well.

\smallskip
Missing proofs can be found in the technical report available at \url{https://arxiv.org/pdf/2204.14261.pdf}.

\subsection*{Related Work}
We next discuss some existing work relevant to our study.

\smallskip \noindent 
\textbf{OMQE of Navigational Queries. }As previously discussed, there exist various works on unrestricted entailment of navigational query languages mediated by DL ontologies. Most of them concentrate on extensions of regular path queries (RPQs), such as  UCRPQs, and consider both Horn~\cite{DBLP:journals/jair/BienvenuOS15} and expressive DLs~\cite{DBLP:journals/iandc/CalvaneseEO14,DBLP:conf/aaai/Gutierrez-Basulto18,DBLP:conf/ijcai/GogaczGIJM19,DBLP:conf/ijcai/BednarczykR19}. There have been also some studies on entailment of graph XPath queries~\cite{DBLP:journals/jair/StefanoniMKR14,DBLP:conf/kr/BienvenuCOS14,DBLP:conf/dlog/KostylevRV14}.

\smallskip \noindent 
\textbf{Finite OMQE. }
There exist various decidability results and optimal complexity bounds for finite entailment of union of conjunctive queries in  Horn DLs~\cite{DBLP:conf/esws/Rosati08,DBLP:conf/kr/GarciaLS14} and in expressive DLs from the $\mathcal S$ family~\cite{GogaczIM18,DBLP:conf/ijcai/GogaczGIJM19,DBLP:conf/mfcs/DanielskiK19}. In most cases, the computational complexity coincides with that of the unrestricted case, but the algorithmic approaches are completely different. On the negative side, undecidability of  finite entailment of UCQs in the more expressive DL $\mathcal{SHOIF}$ was shown by~\cite{DBLP:conf/kr/Rudolph16}, as well as the undecidability result for finite entailment of 2RPQs in $\mathcal{ALCIOF}$. Closer to our work are the positive results  on finite entailment of UCQs with transitive closure over roles  in expressive DLs  allowing for transitivity or transitive closure over roles~\cite{DBLP:conf/mfcs/DanielskiK19,GogaczGGIM20}. These results  close the distance to the undecidability frontier for finite entailment from a different angle by considering ontology languages more expressive than \ALC{}, but a subclass of UCRPQs as query language. In the  context of database theory research, finite OMQE (also called open-world query entailment) has also been investigated; for instance,~\citeauthor{DBLP:journals/tocl/AmarilliB20}~(\citeyear{DBLP:journals/tocl/AmarilliB20}) study finite OMQE  for inclusion dependencies  and functional dependencies over relations of arbitrary arity, and \citeauthor{DBLP:journals/iandc/Pratt-Hartmann09}~(\citeyear{DBLP:journals/iandc/Pratt-Hartmann09}) looks at finite OMQE in the two-variable fragment of FOL with counting quantifiers.


\smallskip \noindent 
\textbf{Finite Controllability. }
There have been also a few works on finite controllability in the context of DLs. For instance, \citeauthor{DBLP:conf/dlog/BednarczykK21}~(\citeyear{DBLP:conf/dlog/BednarczykK21}) recently showed that the $\mathcal{ZOI}$ and $\mathcal{ZOQ}$ members of the $\mathcal{Z}$ family are finitely controllable for UCQs. Beyond DLs, there have been several works on UCQ-finite controllability:  for the guarded fragment of FOL~\cite{DBLP:journals/corr/BaranyGO13} or for various fragments of existential rules~\cite{DBLP:conf/datalog/CiviliR12,DBLP:conf/lics/GogaczM13,BAGET20111620,DBLP:conf/ijcai/AmendolaLM18,DBLP:conf/ijcai/GottlobMP18}. Closer to our study, is the work by \citeauthor{DBLP:conf/kr/FigueiraFB20}~(\citeyear{DBLP:conf/kr/FigueiraFB20}) on the classification of finitely and non-finitely controllable subclasses of CRPQs over ontologies formulated in the guarded-negation fragment of FOL or in the frontier fragment of existential rules. However, no complexity results or algorithms for finite entailment are provided for the non-finitely controllable cases.

\section{Preliminaries}
\label{sec:preliminaries}

\subsection{Description Logics} 

We consider a vocabulary consisting of countably infinite disjoint
sets of \emph{concept names} $\concepts$, 
\emph{role names} $\roles$, and \emph{individual names} $\mn{N_I}$.
\emph{$\ALC$-concepts $C,D$} are defined by the grammar 
\[
  C,D ::= A
  \mid \neg C
  \mid C \sqcap D
  \mid \exists r. C
\]
where $A \in \mn{N_C}$ and  $r \in \mn{N_R}$.  
  We use  standard abbreviations $\bot$, $\top$, $C\sqcup D$ and $\forall r.C$.

An \emph{$\ALC$-TBox $\Tt$} is a finite set of \emph{concept
inclusions (CIs)} $C\sqsubseteq D$, where $C,D$ are
$\ALC$-concepts.  An \emph{ABox} $\Aa$ is a finite non-empty
set of \emph{concept} and \emph{role assertions} of the form $A(a)$,
$r(a,b)$,  where $A \in \mn{N_C}$, $r \in \mn{N_R}$ and
$\{a,b\} \subseteq \mn{N_I}$. \new{We write $\Ind(\Aa)$ for the \emph{set of individual names} occurring in $\Aa$.}
 A \emph{knowledge base (KB)} is a pair
$\Kk=(\Tt, \Aa)$. 
\new{We write $\CN(\Kk)$ and $\Rol(\Kk)$ for the \emph{sets of all concept and role names}  occurring in $\Kk$.}  
We let $\|\Kk\|$ be the total size of the representation of $\Kk$.

Without loss of generality, we assume throughout the paper that all CIs are in one of the following \emph{normal forms}:
\[\bigsqcap_i A_i \sqsubseteq \bigsqcup_j B_j,
  \quad A \sqsubseteq \exists r.B,
  \quad A \sqsubseteq \forall r.B,
  \]
where $A,A_i,B,B_j \in \mn{N_C}$, $r\in \mn{N_R}$,
and empty disjunction and conjunction are equivalent to $\bot$ and
$\top$, respectively. 
Additionally, for each  $A \in \CN(\Kk)$ there is a
complementary $\bar A \in \CN(\Kk)$ axiomatized with $\top
\sqsubseteq A \sqcup \bar A$ and $A \sqcap \bar A \sqsubseteq \bot$.

\subsection{Interpretations} 

The semantics is given as usual via \emph{interpretations} $\Ii=
(\Delta^\Ii, \cdot^\Ii)$ consisting of a non-empty \emph{domain}
$\Delta^\Ii$ and an \emph{interpretation function $\cdot^\Ii$}
mapping concept names to subsets of the domain and role names to
binary relations over the domain, \new{and individual names to elements of the domain.} 
The
interpretation of complex concepts $C$ is defined in the usual
way~\cite{DLBook}. An interpretation $\Ii$ is a
\emph{model of a TBox $\Tt$}, written $\Ii \models\Tt$ if
$C^\Ii \subseteq D^\Ii$ for all CIs $C\sqsubseteq D\in \Tt$.  \new{It is a \emph{model of an ABox $\Aa$}, written $\Ii\models \Aa$,
  if $\mathsf{ind}(\Aa) \subseteq \Delta^\Ii$, $a^\Ii = a$ for each
  $a\in\Ind(\Aa)$, $(a,b)\in r^\Ii$ for all $r(a,b)\in \Aa$, and $a\in A^\Ii$ for
 all $A(a)\in \Aa$.} \new{The first two conditions constitute the so-called \emph{standard name
    assumption.}}
Finally, $\Ii$ is a \emph{model of a KB
$\Kk=(\Tt, \Aa)$}, written $\Ii \models \Kk$, if $\Ii\models \Tt$ and 
$\Ii \models \Aa$.

An interpretation $\Ii$ is \emph{finite} if $\Delta^\Ii$ is finite. An
interpretation $\Ii'$ is a \emph{sub-interpretation} of $\Ii$,
written as $\Ii'\subseteq \Ii$, if $\Delta^{\Ii'}\subseteq
\Delta^\Ii$, $A^{\Ii'}\subseteq A^\Ii$, and $r^{\Ii'}\subseteq
r^{\Ii}$ for all $A\in\mn{N_C}$ and $r\in \mn{N_R}$. For $\Sigma
\subseteq \concepts \cup \roles$, $\Ii$ is an interpretation \emph{over signature} $\Sigma$ if $A^\Ii=\emptyset$ and
$r^\Ii=\emptyset$ for all $A\in \concepts \setminus \Sigma$ and $r\in
\roles\setminus\Sigma$. 
The union $\Ii \cup
\Jj$ of $\Ii$ and $\Jj$ is an interpretation such that $\Delta^{\Ii
\cup \Jj} = \Delta^{\Ii} \cup \Delta^{\Jj}$, $A^{\Ii \cup \Jj} =
A^{\Ii} \cup A^{\Jj}$, and $r^{\Ii \cup \Jj} = r^{\Ii} \cup r^{\Jj}$
for all $A\in\mn{N_C}$ and $r\in \mn{N_R}$.

A \emph{unary $\Kk$-type} is a subset of $\CN(\Kk)$ including either $A$ or $\bar A$ for each $A \in \CN(\Kk)$. 
For an interpretation $\Ii$ and an element $d \in \Delta^\Ii$, the \emph{unary $\Kk$-type of $d$ in $\Ii$} is $\tp^\Ii(d) = \left \{ A \in \CN(\Kk) \bigm | d \in A^\Ii\right\}$. 
We say that $\Ii$ \emph{realizes} a unary $\Kk$-type $\tau$ if $\tau = \tp^\Ii(d)$ for some $d \in \Delta^\Ii$.

\subsection{Queries and Finite Entailment}

We next introduce the query language.
We concentrate on Boolean queries, that is, queries without answer
variables. The extension to queries with answer variables is standard;
see, for example,~\cite{GlimmLHS08}. A \emph{conjunctive 
regular path query (CRPQ)} is a first-order formula
\[\varphi=\exists\mathbf x\, \psi(\mathbf x)\]
such that $\psi(\mathbf x)$ is constructed using $\wedge$ over
atoms of the form $A(t)$ or $\Ee(t,t')$   where $A \in \mathsf{N_C}$, $t,t'$ are variables from $\mathbf x$
or individual names from $\mn{N_I}$, and $\Ee$ is a \emph{path
expression} defined by the grammar
%
$$\Ee,\Ee' ::= r  \mid \Ee^* \mid \Ee \cup \Ee'
\mid \Ee\circ\Ee'$$
where $r\in \roles$. 
Thus, $\Ee$ is essentially
a regular expression over the (infinite) alphabet $\{r\mid r\in \roles \}$. 
The set of individual names in
$\varphi$ is denoted with $\mn{ind}(\varphi)$. A \emph{conjunctive query (CQ)} is a CRPQ that does not use the operators $*,\cup$ and $\circ$ in path expressions, and a \emph{regular path query (RPQ)} consists of a single atom of the form $\Ee(t,t')$.

The semantics of CRPQs is defined via matches. Let us fix a CRPQ
$\varphi=\exists \mathbf x \,\psi(\mathbf x)$ and an interpretation $\Ii$. A
\emph{match for $\varphi$ in $\Ii$} is a function
\[\pi:\mathbf x \cup \mathsf{ind}(\varphi)\to \Delta^\Ii \]
such that $\pi(a)=a$, for all $a\in \mn{ind}(\varphi)$, and
$\Ii,\pi\models\psi(\mathbf x)$ under the standard semantics 
of first-order logic extended with a rule for atoms of the form
 $\Ee(t,t')$. More formally, we define:
\begin{itemize}

 \item $\Ii,\pi\models \psi_1\wedge\psi_2$ iff $\Ii,\pi\models
   \psi_1$ and $\Ii,\pi\models \psi_2$;

\item $\Ii,\pi\models A(t)$ iff $\pi(t) \in A^\Ii$;
  \item $\Ii,\pi\models \Ee(t,t')$ iff
    $(\pi(t),\pi(t'))\in\Ee^\Ii$, where $\Ee^\Ii$ is defined
    inductively as $(\Ee^*)^\Ii = (\Ee^\Ii)^*$, 
    $(\Ee_1\cup \Ee_2)^\Ii  = \Ee_1^\Ii\cup \Ee_2^\Ii$,
    $(\Ee_1\circ \Ee_2)^\Ii  = \Ee_1^\Ii\circ \Ee_2^\Ii$.
    %
\end{itemize}
An interpretation $\Ii$ \emph{satisfies} $\varphi$, written $\Ii\models \varphi$, if
there exists a match for $\varphi$ in $\Ii$. A \emph{union of CRPQs (UCRPQ)} is a finite set of CRPQs and a \emph{union of CQs (UCQ)} is a finite set of CQs. An interpretation $\Ii$ satisfies an UCRPQ $\Phi$, written as $\Ii\models
\Phi$, if $\Ii \models \varphi$ for some $\varphi \in \Phi$. We say that $\Kk$ \emph{finitely entails} $\Phi$, written
$\Kk\fentails \Phi$, if each finite model of $\Kk$ satisfies $\Phi$.  A model
of $\Kk$ that does not satisfy $\Phi$ is a \emph{counter-model}.  The
\emph{finite entailment problem} asks if a given KB $\Kk$ finitely
entails a given query $\Phi$.

\subsection{UCRPQs via Semiautomata}

We work with UCRPQs represented by means of a \emph{semiautomaton} \cite{AlgebraicAutomata} $\Bb = (Q, \Gamma, \delta)$ where $Q$ is a finite set of states, $\Gamma \subseteq \{r\mid r \in \mathsf{N_R}\}$ is a finite alphabet---throughout the paper we assume $\Gamma = \Rol(\Kk)$, and $\delta:Q\times\Gamma \to Q$ is the transition function. A semiautomaton is  essentially a deterministic finite automaton without initial and final states; a run of a semiautomaton $\Bb$ over a word $w$ is defined just like for a finite automaton, except that it can begin in any state and there is no notion of accepting runs. 
Under this representation, an RPQ is an atom over a  binary predicate of the form $\Bb_{q,q'}$ where $q,q' \in Q$ are states of $\Bb$. We let $\Ii,\pi\models
\Bb_{q,q'}(t,t')$ iff  $(\pi(t),\pi(t')) \in \Bb_{q,q'}^\Ii$ where $\Bb_{q,q'}^\Ii$ is the set of pairs $(e,e')$ such that  for some $n\in \mathbb{N}$ there exist
$r_1,\ldots,r_{n} \in \Gamma$ and $e_0,\ldots,e_n\in \Delta^\Ii$
such that 
\begin{itemize} 
\item $e_0=e$ and $e_n=e'$;
\item $(e_{i-1},e_i) \in (r_i)^\Ii$ for all $i\in \{1,\ldots,n\}$;
\item there exists a run of $\Bb$ on the word $r_1\ldots r_{n}$ that begins in state $q$ and ends in state $q'$. 
\end{itemize}
We also allow \emph{edge atoms} of the form $r(x,x')$ for $r\in\Gamma$.

Each UCRPQ $\Phi$ can be effectively rewritten into a UCRPQ $\Phi'$ expressed by means of a semiautomaton $\Bb$ of size $k\cdot 2^{O(m)}$ where $k$ is the number of path expressions in $\Phi$ and $m$ is their maximal size. The size of CRPQs in $\Phi'$ is bounded by the size of CRPQs in $\Phi$ and $|\Phi'| = 2^{\poly{\|\Phi\|}}$, where $\|\Phi\|$ is the total size of $\Phi$. 

For simplicity we work with KBs $\Kk=(\Tt,\Aa)$ where the ABox $\Aa$ is \emph{trivial}; that is, \new{$\Ind(\Aa) = \{a\}$} for some $a\in\mn{N_I}$ and $\Aa$ contains only concept assertions. The general finite entailment problem can be reduced to this special case using the following lemma.

\begin{restatable}{lemma}{lemabox} \label{lem:abox}
Given an oracle for finite entailment for trivial ABoxes, the general finite entailment $(\Tt,\Aa)\fentails \Phi$ can be decided in time $2^{\poly(\|(\Tt,\Aa)\|)\cdot 2^{\poly(\|\Phi\|)}}$ using calls to the oracle for $\Kk'=(\Tt,\Aa')$ and $\Phi'$ consisting of $2^{\poly(\|\Phi\|)}$ CRPQs of linear size over the same semiautomaton as $\Phi$.
\end{restatable}

\subsection{Entailment Modulo Environment}

We solve the entailment problem using a divide-and-conquer approach in which counter-models are decomposed into simpler ones, whose existence is easier to decide. Each level of this recursive procedure will involve certain modifications to the TBox. For complexity reasons we need to pay close attention to these changes, making sure that no blow-up is involved. To make it easier, we generalize the entailment problem by turning the modifications into a separate part of the input, which allows fixing the TBox for the duration of the whole procedure. 
At every level of the recursion, we will need to reason `externally' about the way simpler pieces are put together to form the larger counter-model, and `internally' about how to specify the required properties of a piece depending on what is happening outside. We will think of the models as induced subinterpretations of a larger interpretation. Dually, the remaining part of the larger interpretation can be seen as an external context, in which our models live. The relevant features of this context will be represented by environments, which we now define.

An \emph{environment} $\Ee = (\Theta, \varepsilon)$ consists of a set $\Theta$ of unary types and a function $\varepsilon: \Theta \to 2^{\Rol(\Kk) \times \CN(\Kk)}$. The intended meaning is that only types from $\Theta$ are allowed and each element of an allowed unary type $\tau$ has an $r$-edge to an element in the extension of  $B$ in the external context for each $(r,B)\in \varepsilon(\tau)$. Accordingly, we say that $\Ii$ is a \emph{model of $\Kk$ modulo $\Ee$} and write ${\Ii
  \models^\Ee \Kk}$ if $\Ii$ realizes only unary types from $\Theta$
and it is a model of $\Kk$ under the following
\emph{relaxed semantics of existential restrictions}: 
\begin{itemize}
  \item for every existential restriction $\exists r. B$ in $\Kk$ and every element $d \in \Delta^\Ii$, $d \in (\exists r. B)^\Ii$ iff either there is an $r$-edge in $\Ii$ from $d$ to an element $e \in B^\Ii$ or $(r,B) \in \varepsilon\left(\mathsf{tp}^\Ii(d)\right)$.
\end{itemize}
(The semantics of universal restrictions is not altered and it is the environment's reponsibility to account for them.)
Correspondingly, a query $\Phi$ is \emph{finitely entailed by
$\Kk$ modulo $\Ee$}, written $\Kk \fentails^\Ee \Phi$, if for
each finite interpretation $\Ii$, if $\Ii \models^\Ee \Kk$ then $\Ii
\models \Phi$. The problem of \emph{finite entailment modulo environment} is to
decide for a given KB $\Kk$, environment $\Ee$, and query $\Phi$ if $\Kk
\fentails^\Ee \Phi$.

Note that finite entailment modulo environment and ordinary finite entailment are interreducible. In one direction, it is enough to take the set of all unary $\Kk$-types for $\Theta$ and set $\varepsilon(\tau) = \emptyset$ for all $\tau \in \Theta$. In the other direction, the conditions imposed on unary types and the relaxed semantics of existential restrictions can be expressed easily in the TBox. The latter reduction, however, might significantly increase the size of the TBox. It is easier to control the size of the input at different levels of the recursion when these conditions are explicitly represented in the environment. 

\section{Expansion and Decorations}
\label{ssec:expansion}

In order to handle UCRPQs expressed by means of a semiautomaton $\Bb$ we need to be able to trace runs of $\Bb$ that begin in all possible states, on all infixes of the input word. We achieve this using the following construction.

Let us fix an arbitrary linear order on the set $Q$ of the states of $\Bb$. The \emph{expansion} of $\Bb$ is a semiautomaton $\widehat \Bb$ whose set of states is the set $\widehat Q$ of all permutations of $Q$. Thus, an element of $\widehat Q$ can be seen as a tuple $\pp=(p_1, p_2, \dots, p_n)$ such that $p_i$ is the image of the $i$th state of $\Bb$ under the respective permutation. We refer to positions in this tuple as \emph{levels}. In particular, the \emph{level of $q\in Q$ in $\pp$} is the unique $i$ such that $q = p_i$. Assuming $\delta:Q \times \Rol(\Kk) \to Q$ is the transition function of $\Bb$, we define the transition function  \[\hat \delta : \widehat Q \times \Rol(\Kk) \to \widehat Q\] of $\widehat \Bb$ by letting $\hat \delta \big(\pp, r\big)$ be the permutation $\pp'$ obtained by listing all states appearing in the sequence \[\delta(\pp,r) = \big(\delta(p_1,r),\delta(p_2,r),\dots, \delta(p_n,r)\big)\] in the order of their first appearances, followed by all remaining states of $\Bb$ ordered as in $Q$. Note that the level of $\delta(p_i,r)$ in $\pp'$ is at most $i$.
Consider the set $P \subseteq \{1,2,\dots,n\}$ of levels $i$ such that the level of  $\delta(p_i,r)$ in $\pp'$ is equal to $i$. It follows from the definition of $\pp'$ that $P = \{1, 2, \dots, \ell\}$ for some $\ell \in \{1, 2, \dots, n\}$. We call this number $\ell$ the \emph{level of transition $\pp \stackrel{r}{\longrightarrow} \pp'$}. 

From each run of $\widehat\Bb$ on a word $w$ we can reconstruct all runs of $\Bb$ on $w$. Let $\pp_0,\pp_1,\dots, \pp_m$ be a run of $\widehat \Bb$ on $w$. Consider a run $q_0, q_1, \dots, q_m$ of $\Bb$ on $w$. For $i=0, 1, \dots, m$, let $\ell_i$ be the level of $q_i$ in $\pp_i$. Any sequence $\ell_0, \ell_1, \dots, \ell_m$ associated like this with a run of $\Bb$ will be called a \emph{thread} in the run of $\widehat \Bb$ (see Fig.~\ref{fig:thread}). Notice that two threads that begin at different levels can meet at the same level somewhere along the run; if this happens they remain equal until the end of the run. Also, threads can be born in the middle of a run of $\widehat\Bb$, but they never disappear. 
A crucial property of threads is that they are non-increasing sequences: the level of $q_{i+1}$ in $\pp_{i+1}$ is bounded by the level of  $q_{i}$ in $\pp_{i}$.

\begin{figure*}
    \centering
    \includegraphics[scale=0.6]{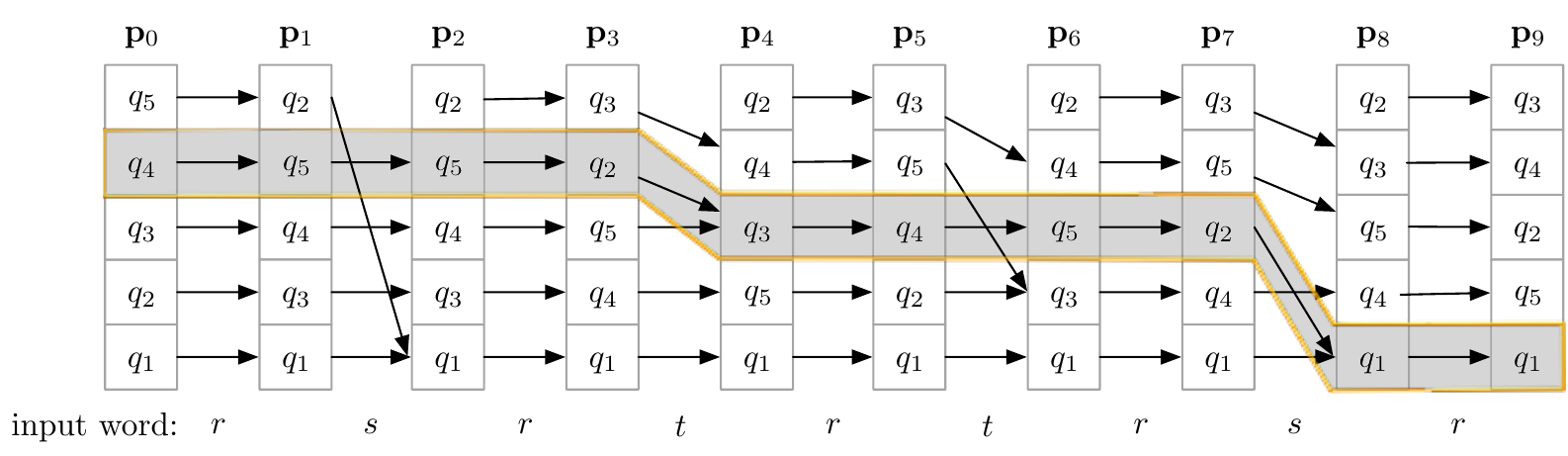}
    \caption{A thread in a run of the expansion of a semiautomaton.}
    \label{fig:thread}
\end{figure*}

\begin{restatable}{lemma}{lemlevels} \label{lem:levels}
Let $\pp_0,\pp_1,\dots, \pp_m$ be a run of $\widehat \Bb$ on $w$, and let $q,q'$ be states of $\Bb$. There is a run of $\Bb$ on $w$ from $q$ to $q'$ iff there exist positions $0 \leq j_1 < j_2 < \dots <  j_k = m$, levels $n \geq \ell_1 > \ell_2 > \dots > \ell_k \geq 1$, and states $q_0, q_1, \dots, q_{k}$ with $1\leq k\leq n$ such that $q_0=q$, $q_k=q'$, and
\begin{itemize}
    \item the level of $q_0$ in $\pp_0$ is $\ell_1$ and
    the level of $q_k$ in $\pp_m$ is $\ell_k$;
    \item for all $i \in \{1, 2, \dots, k-1\}$, the level of $q_i$ in $\pp_{j_i}$ is $\ell_i$ and 
    the level of $\delta\big(q_i, w[j_i+1]\big)$ in $\pp_{j_i+1}$ is $\ell_{i+1}$;
    \item for all $i \in \{1, 2, \dots, k\}$, each transition taken in the segment of the run from $\pp_{j_{i-1}+1}$ (or $\pp_0$ for $i=1$) to $\pp_{j_i}$ has level at least $\ell_i$. 
\end{itemize}
\end{restatable}

As an illustration of Lemma~\ref{lem:levels}, consider the run of the expanded semiautomaton shown in Fig.~\ref{fig:thread}. Tracing the run of the original semiautomaton on the same word, starting in state $q_4$, we discover the positions $j_1=3$ and $j_2=7$ where the corresponding thread drops to a lower level. Between these positions, the thread stays at the same level, beginning with $\ell_1=4$ (taking transitions of levels $5, 4, 5 \geq \ell_1$), followed by $\ell_2=3$ (taking transitions of levels $5, 3, 5\geq \ell_2$), and $\ell_3=1$ (taking a transition of level $5 \geq \ell_3$).

The next step is to account for the possible runs of $\widehat \Bb$ over paths in the interpretation. Towards this goal, we decorate elements of the interpretation with states of $\widehat \Bb$. To avoid additional blow-up, we represent states of $\widehat \Bb$ using combinations of fresh concept names $C_{q,\ell}$ where $q$ is a state of $\Bb$ and $\ell\in \{1, 2, \dots, n\}$ is a level; we write $\conceptsBb$ for the set of all $C_{q,\ell}$. For a state $\pp = (p_1, p_2, \dots, p_n)$ of $\widehat \Bb$, by $C_\pp$ we mean the concept $C_{p_1,1} \sqcap C_{p_2,2} \sqcap \dots \sqcap C_{p_n,n}$. We say that an element $e\in\Delta^\Ii$ is \emph{decorated} with state $\pp$ if $e\in C_\pp^\Ii$. 
An interpretation $\Ii$ is \emph{$\widehat\Bb$-decorated} if no element has incoming edges over different roles from $\Rol(\Kk)$ and $\Ii$ satisfies the CIs \[C_\pp \sqsubseteq \forall r. C_{\hat \delta(\pp,r)}\,,  \quad C_\pp \sqcap C_{\pp'} \sqsubseteq \bot\,, \quad \top \sqsubseteq  \bigsqcup_{\pp\in \widehat Q} C_{\pp}\] for all states $\pp,\pp'$ of $\widehat\Bb$ such that $\pp\neq \pp'$. The axiomatization above is exponential in the size of $\Bb$, but we can do better. 

\begin{restatable}{lemma}{lemaxioma}~\label{lem:axioma}
Given $\Bb$ one can compute in polynomial time a TBox $\widehat \Tt_{\Bb}$ such that $\Ii\models \widehat \Tt_{\Bb}$ iff $\Ii$ is $\widehat \Bb$-decorated.
\end{restatable}


To every edge in a $\widehat\Bb$-decorated interpretation $\Ii$ we can assign a level as follows.  Consider elements $e,e' \in \Delta^\Ii$ such that $(e,e') \in r^\Ii$ for some $r\in\Rol(\Kk)$. Note that $(e,e')\notin s^\Ii$ for every $s \in \Rol(\Kk) \setminus \{r\}$. Let $\pp$ and $\pp'$ be the states decorating $e$ and $e'$, respectively. It holds that $\pp \stackrel{r}{\longrightarrow} \pp'$. By \emph{the level of the edge} $(e,e')$ we shall understand the level of this transition. A \emph{level-$\ell$ interpretation} is a $\widehat \Bb$-decorated interpretation that does not contain edges of level strictly below $\ell$; if $\ell>n$, this means that there are no edges at all. The following lemma is the key to our algorithm. 

\begin{restatable}{lemma}{lemreach}
\label{lem:reach}
Consider a level-$\ell$ interpretation $\Ii$ and elements $e \in C_{q,\ell}^\Ii$ and $e'\in C_{q',\ell}^\Ii$. Then, $(e,e') \in \Bb_{q,q'}^\Ii$ iff there is a path from $e$ to $e'$ in $\Ii$.
\end{restatable}

We make use of Lemma~\ref{lem:reach} by decomposing RPQs into segments corresponding to different levels, as was done for the runs of $\widehat\Bb$ in Lemma~\ref{lem:levels}. To facilitate this, we make our queries aware of levels. 
A \emph{$\widehat\Bb$-decorated CRPQ} is a CRPQ $\varphi$ represented by means of semiautomaton $\Bb$ that contains exactly one atom of the form $C_{q,\ell}(x)$ and exactly one atom of the form $C_{q',\ell'}(x')$ for each atom $\Bb_{q,q'}(x,x')$ in $\varphi$. We call $\ell$ and $\ell'$ the \emph{begin level} and the \emph{end level} of atom $\Bb_{q,q'}(x,x')$ in $\varphi$, respectively. Because levels never increase in a thread of a run of $\widehat \Bb$, we can assume without loss of generality that $\ell \geq \ell'$ always holds. 
%
%
A \emph{level-$\ell$ CRPQ} is a $\widehat \Bb$-decorated CRPQ that contains no RPQ atoms of end level strictly below $\ell$. As all end levels are at most $n$, a  level-$\ell$ CRPQ for $\ell>n$ contains no RPQ atoms; that is, it is a CQ. To \emph{complete} a CRPQ $\varphi$ means to turn it into a $\widehat\Bb$-decorated CRPQ $\varphi'$ by adding unary atoms over concepts $C_{q,\ell}$ in an arbitrary minimal way. Each resulting $\varphi'$ is called a \emph{completion} of $\varphi$. Over $\widehat\Bb$-decorated interpretations, $\varphi$ is equivalent to the union of its completions. The \emph{completion of a UCRPQ} $\Phi$ is the union of all completions of all CRPQs in $\Phi$.

We conclude this section by showing how to turn any counterexample to $\Kk \fentails^\Ee \Phi$ into a $\widehat\Bb$-decorated one.
Let $\Ii$ be an interpretation over $\CN(\Kk) \cup \Rol(\Kk)$. The \emph{product} of $\Ii$ and $\widehat \Bb$ is the interpretation $\Ii \times \widehat \Bb$ over $\conceptsBb\cup \CN(\Kk) \cup \Rol(\Kk)$ such that  
\begin{itemize}
    \item $\Delta^{\Ii\times\widehat\Bb} = \Delta^\Ii \times \Rol(\Kk) \times \widehat Q$,  
    \item $C^{\Ii\times\widehat\Bb} = C^{\Ii} \times \Rol(\Kk) \times \widehat Q$ for all $C \in \CN(\Kk)$,
    \item $C_{q,\ell}^{\Ii\times\widehat\Bb} = \Delta^{\Ii} \times \Rol(\Kk) \times \{(p_1, p_2, \ldots, p_n) \in \widehat Q : p_\ell = q\}$ for all $q\in Q$ and $\ell \in \{1,2,\dots, n\}$,
    \item $r^{\Ii\times\widehat\Bb} = \big\{\big((e,s,\pp), (e',r,\pp')\big): (e,e') \in r^{\Ii}, \pp \stackrel{r}{\longrightarrow} \pp',\linebreak s \in \Rol(\Kk)\big\}$ for $r \in \Rol(\Kk)$.
\end{itemize}
Note that if $\Ii$ is finite, so is $\Ii\times\widehat\Bb$.

\begin{restatable}{lemma}{lemdecoratemodel}\label{lem:decorate_model}
Let $\Phi$ be a UCRPQ, $\Kk$ an $\ALC$ KB with a trivial ABox, and $\Ee$ an environment.
\begin{itemize}
\item $\Ii \times \widehat\Bb$ is a $\widehat\Bb$-decorated interpretation. 
\item If $\Ii \not\models \Phi$ then $\Ii \times \widehat\Bb \not\models \Phi$. 
\item If $\Ii \models^\Ee \Kk$ then $\Ii\times\widehat\Bb \models^\Ee \Kk$ up to identifying the unique individual $a$ in $\Kk$ with some $(a, r, \pp) \in\Delta^{\Ii\times\widehat\Bb}$.
\end{itemize}
\end{restatable}

\section{Core Computational Problem}

To solve the entailment problem we eliminate the lowest level from the query and from the interpretation, and solve the problem with fewer levels recursively. 
%
%
Eliminating each level will involve interpretations built from pieces that are solutions for  the simplified problem.
Evaluating CRPQs over such interpretations requires breaking them down into fragments and it must accommodate single RPQs witnessed across multiple pieces.



For a UCRPQ $\Phi$ let $\tilde\Phi$ be the completion of an equivalent UCRPQ represented by means of a semiautomaton $\Bb$.
A~\emph{fragment} of $\varphi\in\tilde\Phi$ is either of the following:
\begin{itemize}
\item a $\widehat \Bb$-decorated CRPQ of the form $C_{q_1,\ell_1}(y_1) \land \Bb_{q_1,q_2}(y_1,y_2) \land C_{q_2,\ell_2}(y_2)$ or $C_{q_1,\ell_1}(y_1) \land \Bb_{q_1,q_2}(y_1,y_2) \land C_{q_2,\ell_2}(y_2) \land r(y_2,y_3) \land C_{q_3,\ell_3}(y_3)$
where $y_1,y_2,y_3$ are fresh variables and $r\in\Rol(\Kk)$, 
\item a connected $\widehat \Bb$-decorated CRPQ that can be obtained from $\varphi$ by dropping selected atoms, replacing selected RPQ atoms $\Bb_{q,q'}(x,x')$ by a subset of  $\Bb_{q,q_1}(x,y_1)$,  $r(y_1,y_2)$, $\Bb_{q_3,q'}(y_3,x')$ for some fresh variables $y_1,y_2,y_3$ 
and \mbox{$r\in\Rol(\Kk)$}, and  completing the resulting CRPQ.
\end{itemize}
A fragment of $\Phi$ is a fragment of any of the CRPQs in $\tilde\Phi$. 
Importantly, a fragment of a fragment of $\Phi$ is also a fragment of $\Phi$, and each $\varphi\in\tilde\Phi$ is a fragment of $\Phi$.
Up to renaming fresh variables, $\Phi$ has $2^{\poly(\|\Phi\|)}$ different fragments, despite $\Bb$ being exponential in $\|\Phi\|$.


We now enrich interpretations again by including information about matched fragments of $\Phi$. For each fragment $\varphi$ of $\Phi$ and each $\emptyset \neq V \subseteq \var(\varphi)$ we choose a fresh concept name $A_{\varphi, V}$.
We call an interpretation $\Ii$ \emph{correct  (wrt.~$\Phi$)} if $e \in A_{\varphi, V}^{\Ii}$ iff $\pi(V) = \{e\}$ for some match $\pi$ for $\varphi$ in $\Ii$. Assuming $\Ii$ is correct, $\Ii \models \Phi$ iff $A_{\varphi, V}^{\Ii} \neq \emptyset$ for some $\varphi\in \tilde\Phi$ and $\emptyset \neq V \subseteq \var(\varphi)$. Correctness is not compositional: the union of two correct interpretations sharing a single element need not be correct. As our method of eliminating levels relies on such decompositions of interpretations,  we replace correctness with a notion that is weaker, but compositional. 

We first abstract the decomposition of a $\widehat\Bb$-decorated CRPQ induced by a match in a union of disjoint `peripheric' interpretations, each sharing a single element with a single `core' interpretation (Fig.~\ref{fig:partition} shows three `peripheric' interpretations connected to the `core'  by single edges, included in the `peripheric' interpretations).
A \emph{partition} of a $\widehat\Bb$-decorated CRPQ $\varphi$ into $\varphi', \varphi_1,\dots, \varphi_k$ is obtained as follows. Choose $X', X_1, \dots X_k \subseteq \var(\varphi)$ such that 
\begin{itemize}
\item $X_i \cap X_j = \emptyset$ for all $i \neq j$;
\item for each atom of the form $r(x,x')$ in $\varphi$ there exists $i$ such that either $\{x,x'\} \subseteq X_i$ or $\{x,x'\} \subseteq X'$.
\end{itemize}
Based on $X', X_1, \dots X_k$ define $\varphi', \varphi_1, \dots, \varphi_k$ as follows.  For each atom of the form $r(x,x')$  in $\varphi$ choose $i$ such that  $\{x,x'\} \subseteq X_i$ and add $r(x,x')$ to $\varphi_i$ or add $r(x,x')$ to $\varphi'$ provided that $\{x,x'\}\subseteq X'$. For each RPQ atom  $\Bb_{q,q'}(x,x')$ of begin level $\ell$ and end level $\ell'$ in $\varphi$ do one of the following:
\begin{itemize}
    \item provided that $\{x,x'\} \subseteq X'$, add $\Bb_{q,q'}(x,x')$ to $\varphi'$;
    \item choose $i$ such that $\{x,x'\} \subseteq X_i$ but  $\{x,x'\} \not\subseteq X'$, and add $\Bb_{q,q'}(x,x')$ to $\varphi_i$ (light green RPQ in Fig.~\ref{fig:partition});
    \item choose $i$ such that $x \in X'\setminus X_i$ and $x'\in X_i \setminus X'$, 
    a level $m$ such that $\ell \geq m \geq \ell'$, a state $p$ of $\Bb$, and a fresh variable $y$, and add $\Bb_{q,p}(x,y) \land C_{p,m}(y)$ to $\varphi'$ and $C_{p,m}(y) \land \Bb_{p,q'}(y,x')$ to $\varphi_i$ (blue and orange in Fig.~\ref{fig:partition});
    \item choose $i$ such that $x \in X_i \setminus X'$ and $x'\in X' \setminus X_i$, 
    a level $m$ such that $\ell \geq m \geq \ell'$, a state $p$ of $\Bb$, and a fresh variable $y$, and add $\Bb_{q,p}(x,y) \land C_{p,m}(y)$ to $\varphi_i$ and $C_{p,m}(y) \land \Bb_{p,q'}(y,x')$ to $\varphi'$ (dark green in Fig.~\ref{fig:partition}); 
    \item choose $i\neq j$ such that $x \in X_i \setminus X'$ and $x'\in X_j \setminus X'$, 
    levels $m, m'$ such that $\ell \geq m \geq m'\geq \ell'$, states $p,p'$ of $\Bb$, and fresh variables $y,y'$, add $\Bb_{q,p}(x,y) \land C_{p,m}(y)$ to $\varphi_i$, $C_{p,m}(y) \land \Bb_{p,p'}(y,y') \land C_{p',m'}(y')$ to $\varphi'$, and $C_{p',m'}(y') \land \Bb_{p',q'}(y',x')$ to $\varphi_j$ (purple in Fig.~\ref{fig:partition}).
\end{itemize}
Note that for each $\Bb_{q,q'}(x,x')$ exactly one of the above actions can be performed and the choice of $i$ and $j$ is unique. To complete the construction, add to $\varphi'$ all unary atoms of $\varphi$ over variables already used in $\varphi'$, and similarly for each $\varphi_i$.
Observe that for each $X' \subseteq \var(\varphi)$ there is exactly one choice of $X_1, X_2, \dots, X_k$ (up to a permutation) such that the resulting $\varphi_1, \varphi_2, \dots, \varphi_k$ are connected (regardless of the choice of $p,p'$ and $m,m'$). Assuming that $\varphi$ is a fragment of $\Phi$, it then holds that so are $\varphi_1, \varphi_2, \dots, \varphi_k$.

\begin{figure}
    \centering
    \includegraphics[width=\columnwidth]{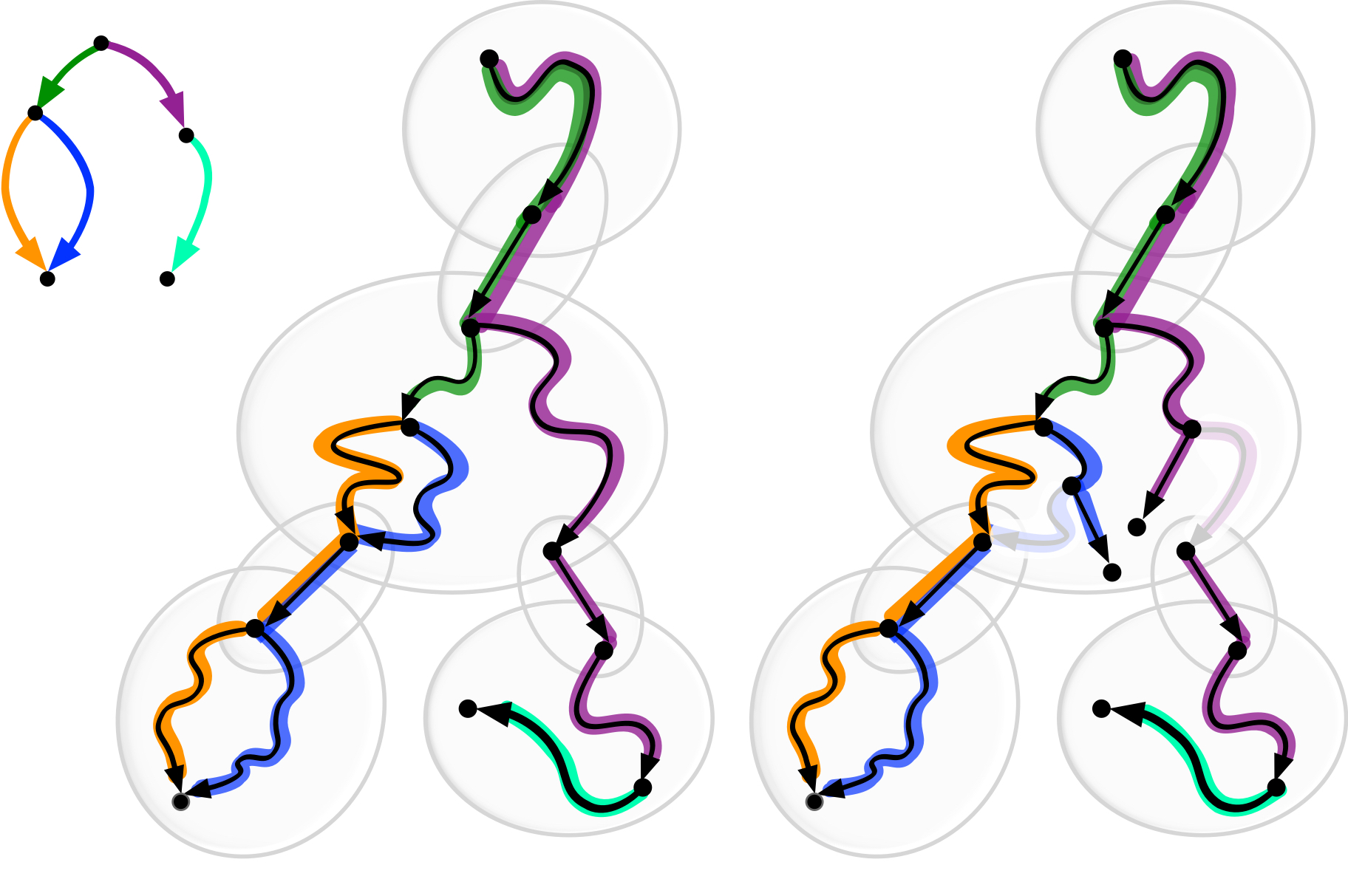}
    \caption{CRPQ $\varphi$ is distributed over the bags constituting $\Ii$.}
    \label{fig:partition}
\end{figure}

We call $\Ii$ \emph{consistent (wrt.~$\Phi$)} if for each partition of a fragment $\varphi$ of $\Phi$ into a CRPQ $\varphi'$ and fragments $\varphi_1, \varphi_2, \dots, \varphi_k$ with $\var(\varphi_i)\cap \var(\varphi_j) = \emptyset$ for $i\neq j$, $V_i = \var(\varphi_i) \cap \var(\varphi')$, and $\emptyset \neq V\subseteq \var(\varphi) \cap \var(\varphi')$, there is no match $\pi$ for $\varphi'$ in $\Ii$ such that $\pi(V_i) = \{e_i\} \subseteq \big(A_{\varphi_i, V_i}\big)^{\Ii}$ for all $i$ but $\pi(V) = \{e\} \not\subseteq \big(A_{\varphi, V}\big)^{\Ii}$. Clearly, all correct interpretations are consistent. The converse is not true in general, but the following key property is preserved.  

\begin{restatable}{lemma}{lemconsistent} \label{lem:consistent}
  For every UCRPQ $\Phi$ and every consistent \mbox{$\widehat\Bb$-decorated} interpretation $\Ii$, if  $A_{\varphi, V}^{\Ii} = \emptyset$ for each $\varphi \in \tilde\Phi$ and $\emptyset \neq V \subseteq \var(\varphi)$, then $\Ii \not\models \Phi$. 
\end{restatable}


Consistency is sufficient to express entailment, but it does not yield well to the recursive elimination of levels. We generalize it by refining the information  about matched fragments of $\Phi$. We introduce fresh concepts $A_{\varphi,V}^\kappa$ where $\varphi$ is a fragment of $\Phi$, $\emptyset \neq V \subseteq \var(\varphi)$,  \[\kappa:\var(\varphi) \to \{1, 2, \dots, \ell\}\,,\]
and $\kappa(V) = \{\ell\}$ for some $\ell\in\{1, 2, \dots, n+1\}$.
We write $\CN_{\ell}^\Phi$ for the set of  $A_{\psi,V}^\kappa$ such that $\kappa(V) = \{\ell\}$.
Intuitively, $\kappa$ is a synopsis of when specific fragments of $\psi$ were matched during the recursive search for the model. Specifically, $\kappa(x) = \ell$ indicates that $x$ was matched after all levels strictly below $\ell$ had been eliminated from the query, but while level $\ell$ was still present. 
Accordingly, $\ell$-consistency, defined below, ensures that the synopses built so far are consistently updated while level $\ell$ is being handled.

We call $\Ii$ \emph{$\ell$-consistent (wrt.~$\Phi$)} if for each partition of a fragment $\varphi$ of $\Phi$ into a CRPQ $\varphi'$ of level $\ell$ and fragments $\varphi_1, \varphi_2, \dots, \varphi_k$ with $\var(\varphi_i)\cap \var(\varphi_j) = \emptyset$ for $i\neq j$, $V_i = \var(\varphi_i) \cap \var(\varphi')$, and $\emptyset \neq V\subseteq \var(\varphi) \cap \var(\varphi')$, there is no match $\pi$ for $\varphi'$ in $\Ii$ such that $\pi(V_i) = \{e_i\} \subseteq \big(A_{\varphi_i, V_i}^{\kappa_i}\big)^{\Ii}$ for all $i$ but $\pi(V) = \{e\} \not\subseteq (A_{\varphi, V}^\kappa)^{\Ii}$ where
\begin{itemize}
    \item $\kappa_i(x) \leq \ell$ for all $x \in  \var(\varphi_i)\,$,
    \item $\kappa(x) = \kappa_i(x)$ for all $x \in  \var(\varphi_i) \setminus V_i\,$,
    \item $\kappa (x) = \ell$ for all $x \in \var(\varphi) \cap \var(\varphi')\,$.
\end{itemize}
%
%
We stress that  while $\varphi'$ has level $\ell$,  fragments $\varphi, \varphi_1,  \dots, \varphi_k$ can have any level. 
Note also that $\ell$-consistency speaks only of concept names in $\CN_1^\Phi \cup\CN_2^\Phi \cup \dots \cup \CN_\ell^\Phi$.
Identifying $A_{\varphi,V}$ with $A_{\varphi,V}^\kappa$ for $\kappa$ constantly equal to 1, we get that  consistency and $1$-consistency are equivalent.

In what follows, by an \emph{$(\ell,\ell')$-interpretation} we mean an $\ell'$-consistent level-$\ell$ interpretation. By an \emph{$(\ell,\ell')$-model} of $\Kk$ modulo $\Ee$ we mean an $(\ell,\ell')$-interpretation that is model of $\Kk$ modulo $\Ee$.
%
%
The actual problem we will be solving is the following  \emph{$(\ell, \ell')$-model problem} for $\ell\leq \ell'$: Given a KB $\Kk$ with a trivial ABox, an environment $\Ee$, and a UCRPQ $\Phi$ decide if there exists a finite $(\ell,\ell')$-model of $\Kk$ modulo $\Ee$.


By Lemma~\ref{lem:decorate_model}, entailment modulo environment (with trivial ABox) can be reduced to the $(1,1)$-model problem by modifying the environment to forbid all unary types containing $A^\kappa_{\varphi, V}$ for any $\varphi \in \tilde\Phi$, $\emptyset \neq V \subseteq \var(\varphi)$, and $\kappa$ constantly equal 1. 
Note that the reduction does not affect the query $\Phi$, nor the KB $\Kk$. However, it introduces up to $2^{\poly(\|\Phi\|)}$ new concept names $A_{\varphi,V}$ and $C_{q,\ell}$. Consequently, the number of unary types is at most $2^{|\CN(\Kk)|+2^{\poly(\|\Phi\|)}}$. It follows that the size of the environment is bounded by $2^{\|\Kk\|+2^{\poly(\|\Phi\|)}}$.

To solve the $(1,1)$-model problem we will proceed recursively, incrementing $\ell$ and $\ell'$ in an alternating fashion, until $\ell=\ell'=n+1$. At each level of the recursion we will be making multiple recursive calls. During the recursion the UCRPQ $\Phi$ and the TBox $\Tt$ will remain unchanged, but the ABox and the environment will evolve. Importantly, we will not introduce any new concepts, so the size of the environment will always be bounded by $2^{\|\Kk\|+2^{\poly(\|\Phi\|)}}$. The size of the ABox will be bounded by $\|\Kk\| + 2^{\poly(\|\Phi\|)}$ and the number of individuals will never grow. 
In consequence, the total cost of the algorithm can be computed as the cost of a single recursion step times the number of steps. In the following sections we will show that each recursion step can be carried out in time 
$2^{O(\|\Kk\|) +  2^{\poly(\|\Phi\|)}}$, excluding the cost of the recursive calls. The depth of the recursion is $O(n) = 2^{\poly(\|\Phi\|)}$. The number of recursive calls within a single recursion step is also bounded by $2^{O(\|\Kk\|) +  2^{\poly(\|\Phi\|)}}$. This means that the total number of recursion steps is $2^{\|\Kk\|\cdot 2^{\poly(\|\Phi\|)}}$ and so is the overall complexity of the recursive algorithm for the $(1,1)$-model problem.

\section{Incrementing the Level of Queries}
\label{ssec:queries}

The main goal of this section is to solve the $(\ell,\ell)$-model problem by reduction to multiple instances of the $(\ell, \ell+1)$-model problem for $\ell\leq n$. The $(n+1,n+1)$-model problem is discussed briefly at the end of the section. 

As a first step, we observe that it is enough to consider interpretations whose DAG of strongly connected components is a tree. For this purpose we define \emph{tree-like} interpretations as those that can be decomposed into multiple finite subinterpretations, called \emph{bags}, arranged into a (possibly infinite) tree such that:
(1) all bags are pairwise disjoint;
(2) between each parent and child bag there is a single edge, pointing from an element of the parent bag to an element of the child bag;  (3) all other edges are between elements of the same bag.
We think of edges between bags as 2-element interpretations, called \emph{edge-bags}, sharing the origin  with the parent bag and the target  with the child bag. 
A tree-like interpretation is then a union of all its bags and edge-bags.
Fig.~\ref{fig:partition} shows a tree-like interpretation with 4 bags and 3 edge-bags.
In tree-like interpretations $\ell$-consistency is a local property. 

\begin{restatable}{lemma}{lemconsistentcomp} \label{lem:consistentcomp}
  A tree-like interpretation is $\ell$-consistent iff each of its bags and edge-bags is $\ell$-consistent.
\end{restatable}



\begin{restatable}{lemma}{lemunravellingconnected} 
\label{lem:unravelling-connected}
There is a finite $(\ell,\ell')$-model of $\Kk$ modulo  $\Ee$ iff there is a finite  tree-like $(\ell,\ell')$-model of $\Kk$ modulo $\Ee$ whose bags are strongly connected.
\end{restatable}

The next step is to eliminate the lowest level from the queries. An \emph{$(\ell+1)$-reduct} of a level-$\ell$ CRPQ $\varphi$ is any CRPQ that can be obtained from $\varphi$ by first splitting each RPQ atom $\Bb_{q_1,q_2}(x_1,x_2)$ of begin level $\ell_1 > \ell$ and end level $\ell$ into  $\Bb_{q_1,q'_1}(x_1,x'_1) \land C_{q'_1, \ell'_1}(x'_1) \land r(x'_1,x'_2) \land C_{q'_2,\ell}(x'_2) \land \Bb_{q'_2,q_2}(x'_2,x_2)$ where $\ell_1 \geq \ell'_1 \geq \ell+1$, and then dropping from the resulting CRPQ all atoms whose begin and end level is $\ell$ (all unary atoms are kept). Note that each $(\ell+1)$-reduct $\varphi'$ of $\varphi$ is a conjunction of at most $|\varphi|$ disjoint fragments of $\varphi$ and that $\var(\varphi) \subseteq \var(\varphi')$.

\begin{restatable}{lemma}{crpqSCCreducts} 
\label{lem:reducts}
Over $\widehat\Bb$-decorated  interpretations, each level-$\ell$ CRPQ implies the union of its $(\ell+1)$-reducts. 
Over strongly-connected level-$\ell$ interpretations, each level-$\ell$ CRPQ is equivalent to the union of its $(\ell+1)$-reducts. 
\end{restatable}


Because Lemma~\ref{lem:unravelling-connected} guarantees tree-like solutions with strongly connected bags, we can replace $\ell$-consistency with \emph{strong $\ell$-consistency}: the only difference is that $\pi$ ranges over matches of all possible $(\ell+1)$-reducts of $\varphi'$, rather than over matches of $\varphi'$ itself. We restate Lemma~\ref{lem:unravelling-connected} as follows.


\begin{restatable}{lemma}{lemstronglyconsistent} \label{lem:strongly-consistent}
There is a finite $(\ell,\ell)$-model of $\Kk$ modulo $\Ee$  iff 
there is a finite tree-like level-$\ell$ model of $\Kk$ modulo $\Ee$ whose edge bags are $\ell$-consistent and bags strongly $\ell$-consistent. 
\end{restatable}

It remains to show how to find models of the latter form. Let us first see how to find one consisting of a single bag; that is, how to find a finite strongly $\ell$-consistent level-$\ell$ model of $\Kk$ modulo $\Ee$. We will show that this amounts to finding a finite  $(\ell, \ell+1)$-model of $\Kk$ modulo $\Ee'$ for one of the $(\ell+1)$-reducts $\Ee'$ of $\Ee$ described below. 



Consider a fragment $\varphi$, a non-empty set $V\subseteq \var(\varphi)$, a partition of $\varphi$ into a CRPQ $\varphi'$ of level $\ell$, and fragments $\varphi_1, \varphi_2, \dots, \varphi_k$, as in the definition of (strong) $\ell$-consistency. Let $\kappa :\var(\varphi) \to \{1,2, \dots, \ell\}$ be such that  $\kappa\big(\var(\varphi)\cap \var(\varphi')\big) = \{\ell\}$.
Let $\psi'$ be an $(\ell+1)$-reduct of $\varphi'$. Consider CRPQs $\psi$ with $\var(\varphi) \subseteq \var(\psi)$ that can be partitioned into $\psi'$ and $\varphi_1, \varphi_2, \dots, \varphi_k$. Choose the one with minimal $\var(\psi)$. This amounts to merging back all RPQ atoms split during the partition of $\varphi$, provided that their segments were not affected by replacing $\varphi'$ with $\psi'$.
The CRPQ $\psi$ is not a fragment, because it need not be connected: 
Figure~\ref{fig:partition} right illustrates passing from $\varphi$ to $\psi$ consisting of two disconnected fragments. Let $\psi_1, \psi_2, \dots, \psi_m$ be the fragments constituting $\psi$ and let $U_i = V \cap \var(\psi_i)$. 
An \emph{$(\ell+1)$-reduct} $\Ee'$ of $\Ee$ is constructed by iterating over all possible choices of $\varphi$, $V$, $\varphi'$, $\varphi_1, \varphi_2, \dots, \varphi_k$, $\psi'$, $\kappa$, as above, and pruning $\Ee$ for each choice in one of the following ways:
\begin{itemize}
\item either pick $i$ such that $U_i = \emptyset$ and remove  all unary types that contain any $A^{\kappa_i}_{\psi_i,W_i}$ with $W_i\subseteq \var(\psi_i) \cap \var(\psi')$, 
$\kappa_i\big(\var(\psi_i) \cap \var(\psi')\big)=\{\ell+1\}$, and 
$\kappa_i(x)=\kappa(x)$ for all $x \in \var(\psi_i)\setminus \var(\psi')$ ;
\item  or remove  all unary types 
that contain some $A^{\kappa_i}_{\psi_i,U_i}$ with
$\kappa_i\big(\var(\psi_i) \cap \var(\psi')\big)=\{\ell+1\}$ and $\kappa(x)=\kappa_i(x)$ for all $x \in \var(\psi_i)\setminus \var(\psi')$,
for each $i$ such that $U_i \neq \emptyset$, but do not contain $A^\kappa_{\varphi,V}$.
\end{itemize}

\begin{restatable}{lemma}{lemenvreduct}
\label{lem:envreduct} 
$\Ii$ is a strongly $\ell$-consistent level-$\ell$ model of $\Kk$ modulo $\Ee$ iff some interpretation that agrees with $\Ii$ over all role names and all concept names except $\CN_{\ell+1}^\Phi$ is an $(\ell+1)$-consistent level-$\ell$ model of $\Kk$ modulo $\Ee'$ for some $(\ell+1)$-reduct $\Ee'$ of $\Ee$.
\end{restatable}

Finite models consisting of multiple bags can be constructed bottom-up by a least fixed point procedure, using Lemma~\ref{lem:envreduct} to find each bag.

\begin{restatable}{lemma}{lemalgoconnected}
\label{lem:algo-connected}
The $(\ell,\ell)$-model problem for an $\ALC$ KB $\Kk$, a UCRPQ $\Phi$, and an environment $\Ee$ can be solved in time \[2^{O(\|\Kk\|)+ 2^{\poly(\|\Phi\|)}}\]  given an oracle for the $(\ell, \ell+1)$-model problem (with the same UCRPQ and TBox).
\end{restatable}

At the bottom of the recursion we need to check if there exists a $(n+1, n+1)$-model for $\Kk$ modulo $\Ee$. Now, a $\widehat\Bb$-decorated interpretation is level-$(n+1)$ iff it is \emph{discrete}; that is, it has no edges at all. This allows solving the problem by a direct inspection. Because the ABox is trivial and $\ell$-consistency is preserved under restrictions of the domain, it is enough to go through all singleton interpretations. 

\begin{restatable}{lemma}{lemaldiscrete}
\label{lem:algo-discrete}
The $(n+1,n+1)$-model problem for an $\ALC$ KB $\Kk$, a UCRPQ $\Phi$, and an environment $\Ee$ can be solved in time $2^{O(\|\Kk\|)+ 2^{\poly(\|\Phi\|)}}$.
\end{restatable}

\section{Incrementing the Level of Models}
\label{ssec:interpretations}

In this section we solve  the $(\ell,\ell')$-model problem by reduction to multiple instances of the $(\ell',\ell')$-model problem for $\ell<\ell'$; that is, we eliminate level-$\ell$ edges from the interpretations. Like in Section~\ref{ssec:queries}, we rely on tree-like models of a  special form; this time, however, they may be infinite and an additional step is needed to turn them into finite ones. 

A $\widehat\Bb$-decorated interpretation is 
\emph{$\ell'$-flat} if it is a tree-like interpretation where all edges between bags have level strictly below $\ell'$, whereas all edges inside bags have level at least $\ell'$. 

\begin{restatable}{lemma}{lemunravellingflat}
\label{lem:unravelling-flat}
If there exists a finite $(\ell,\ell')$-model of $\Kk$ modulo $\Ee$ then there exists
an $\ell'$-flat $(\ell,\ell')$-model of $\Kk$ modulo $\Ee$ with bounded degree and bag size.
\end{restatable}

In contrast to Lemma~\ref{lem:unravelling-connected}, the above only shows that the reformulated condition is necessary. We show that it is sufficient, by turning an arbitrary
$\ell'$-flat $(\ell,\ell')$-model of $\Kk$ modulo $\Ee$ with bounded degree and bag size into a finite $(\ell,\ell')$-model of $\Kk$ modulo $\Ee$. For this we use \emph{coloured blocking}.
For $d \in\Delta^\Ii$, the \emph{$m$-neighbourhood $N_m^{\Ii}(d)$ of $d$} is the interpretation obtained by restricting $\Ii$ to elements $e \in \Delta^\Ii$ within distance $m$ from $d$ in $\Ii$, enriched with a fresh concept interpreted as $\{d\}$.
A \emph{colouring of $\Ii$ with $k$ colours} is an extension $\Ii'$ of $\Ii$ to $k$ fresh concept names $B_1, \dots, B_k$ such that $B_1^{\Ii'}, \dots, B_k^{\Ii'}$ is a partition of $\Delta^{\Ii'} = \Delta^{\Ii}$. We say that $d \in B_i^{\Ii'}$ has colour $B_i$. We call $\Ii'$ \emph{$m$-proper} if for each $d \in \Delta^{\Ii'}$ all elements of $N_m^{\Ii'}(d)$ have different colours. 

\begin{fact}[\protect\citeauthor{GogaczIM18} \protect\citeyear{GogaczIM18}]
  \label{fact:coloured-blocking}
If $\Ii$ has bounded degree, then for all $m\geq 0$ there exists an $m$-proper colouring $\Ii'$ of $\Ii$ with finitely many colours.  
Consider interpretation $\Jj$ obtained from $\Ii'$ by redirecting some edges such that the old target and the new target have isomorphic $m$-neighbourhoods in $\Ii'$. Then, for each conjunctive query $\varphi$ with at most $\sqrt{m}$ binary atoms, if $\Ii\not\models \varphi$, then
$\Jj\not\models \varphi$.  
\end{fact}

Let $\Ii$ be an $\ell'$-flat $(\ell,\ell')$-model of $\Kk$ modulo $\Ee$ of bounded degree with bags of size at most $M$. 
In order to make Fact~\ref{fact:coloured-blocking} applicable, we need to express the $\ell'$-consistency condition over $\Ii$ by means of a finite set of conjunctive queries, rather than CRPQs. Towards this end, we show that over $\Ii$ each level-$\ell'$ CRPQ 
is equivalent to a UCQ. 
We rely on the following observation. 

\begin{restatable}{lemma}{lemboundedpaths}
\label{lem:bounded-paths}
In a match of a $\widehat\Bb$-decorated CRPQ in a $\widehat\Bb$-decorated interpretation, each path witnessing an RPQ atom of end level at least $\ell'$ uses at most  $n - \ell'$ edges of level strictly below $\ell'$.
\end{restatable}


We say that a CRPQ $\varphi$ is \emph{bounded by $K$} over an interpretation $\Jj$ if for each match of $\varphi$ in $\Jj$ each RPQ atom of $\varphi$ can be witnessed by a path of length at most $K$.

\begin{restatable}{lemma}{lemboundedcrpq}
\label{lem:bounded-crpq}
Let $\Jj$ be $\widehat\Bb$-decorated interpretation made up of disjoint level-$\ell'$ interpretations of size at most $M$ connected by edges of level strictly below $\ell'$. Assuming $\ell' \leq n$, each level-$\ell'$ CRPQ is bounded by $M (n-\ell'+1)^2$ over $\Jj$.
\end{restatable}


For a $\widehat\Bb$-decorated CRPQ $\varphi$, let $\varphi^{(K)}$ be the UCQ obtained by taking the union of all CQs that can be obtained from $\varphi$ by eliminating each RPQ atom $\Bb_{q,q'}(x, x')$ in one of the following ways: either remove the atom and equate variables $x$ and $x'$, or replace the atom with a CQ of the form 
\[ r_1(x, y_1) \land r_2(y_1,y_2)\land \dots \land r_N(y_{N-1},x') \]
where $N \leq K$, $y_1, \dots, y_{N-1}$ are fresh variables, 
and  there is a run of $\Bb$ on $r_1\dots r_N$ that  begins in $q$ and ends in $q'$.


\begin{fact} \label{fact:bounded}
If a $\widehat\Bb$-decorated CRPQ $\varphi$ is bounded by $K$ on an interpretation $\Jj$, then $\Jj \models \varphi$ iff $\Jj \models \varphi^{(K)}$.
\end{fact}


The final step before we can apply Fact~\ref{fact:coloured-blocking} is to express $\ell'$-consistency as query evaluation. 
Consider a partition of a fragment $\varphi$ of $\Phi$ into a CRPQ $\varphi'$ of level $\ell'$ and fragments $\varphi_1, \varphi_2, \dots, \varphi_k$ with $\var(\varphi_i)\cap \var(\varphi_j) = \emptyset$ for $i\neq j$, $V_i = \var(\varphi_i) \cap \var(\varphi')$, and $\emptyset \neq V\subseteq \var(\varphi) \cap \var(\varphi')$.
Let $\psi$ be the CRPQ obtained from $\varphi'$ as follows. Begin from a copy of $\varphi'$. 
For each $i \in \{1, \dots, k\}$, add to $\psi$ an atom $A_{\varphi_i, V_i}^{\kappa_i}(u)$ for some $\kappa_i$ satisfying $\kappa_i (x) \leq \ell$ for all $x \in \var(\varphi_i)$ and $\kappa_i(x)=\ell$ for all $x \in V_i$, 
and some variable $u$ in $V_i$ ($V_i$ is nonempty, because $\varphi$ is connected), and equate all variables in $V_i$. Similarly, add to $\psi$ the atom $\bar{A}_{\varphi', V}(u)$ 
for some $\kappa$ satisfying $\kappa (x) = \ell$ for all $x \in \var(\varphi) \cap \var(\varphi')$ and $\kappa(x) = \kappa_i(x)$ for all $x \in \var(\varphi) \cap \var(\varphi_i)$,
and some $u\in V$, 
and equate all variables in $V$. 
Let $\Phi_{\ell'}$ be the union of all CRPQs $\psi$  obtained as above for different choices of $\varphi$, $\varphi'$, $\varphi_1, \varphi_2, \dots, \varphi_k$, $V$, and $\kappa_1, \kappa_2, \dots, \kappa_k$.
Note that $\Phi_{\ell'}$ is a union of  level-$\ell'$ CRPQs. If $\ell'>n$, $\Phi_{\ell'}$ is a UCQ. 

\begin{restatable}{lemma}{lemconsistencyaseval}
\label{lem:consistency-as-evaluation}
If $\Jj$ is a $\widehat \Bb$-decorated interpretation, then $\Jj$ is $\ell'$-consistent iff $\Jj \not\models \Phi_{\ell'}$.
\end{restatable}

Let $K=M(n-\ell'+1)^2$. Let $t$ be the maximal number of binary atoms in one CQ in $\Phi_{\ell'}^{(K)}$.
(Note that if $\ell'>n$, the query $\Phi_{\ell'}$ is a UCQ and $\Phi_{\ell'}^{(K)}$ coincides with $\Phi_{\ell'}$.)
Fix $m = t^2$ and let $\Ii'$ be an $m$-proper colouring of $\Ii$. On each infinite branch, select the first bag $\Mm$ such that for some bag $\Mm'$ higher on this branch, the $m$-neighbourhood of the target element $e$ of the edge from the parent of $\Mm$ to $\Mm$ is isomorphic to the $m$-neighbourhood of the target $e'$ of the edge from the parent of $\Mm'$ to $\Mm'$. Because the number of non-isomorphic $m$-neighbourhoods in a structure of bounded degree is bounded, the depth of the selected bags in the tree of bags is also bounded. The set of selected bags is finite and forms a maximal antichain.  Let $\Ff$ be the interpretation obtained by taking the union of all strict ancestors of the selected bags, and for each element $e$ as above, redirect the edge coming from the parent of $\Mm$ to $e'$.

Clearly, $\Ff$ is a finite level-$\ell$ interpretation. It is routine to check that $\Ff\models^\Ee \Kk$. It remains to prove that $\Ff$ is $\ell'$-consistent. We know that $\Ii$ is $\ell'$-consistent. By Lemma~\ref{lem:consistency-as-evaluation}, $\Ii \not \models \Phi_{\ell'}$. By Lemma~\ref{lem:bounded-crpq} and Fact~\ref{fact:bounded}, $\Ii \not \models \Phi_{\ell'}^{(K)}$. By Fact~\ref{fact:coloured-blocking}, $\Ff \not\models \Phi_{\ell'}^{(K)}$. By construction, $\Ff$ satisfies the assumptions of Lemma~\ref{lem:bounded-crpq}. Hence, by Lemma~\ref{lem:bounded-crpq} and Fact~\ref{fact:bounded}, $\Ff \not\models \Phi_{\ell'}$. 
We conclude that $\Ff$ is $\ell'$-consistent using Lemma~\ref{lem:consistency-as-evaluation}.

Thus we have proved the converse of Lemma~\ref{lem:unravelling-flat}.

\begin{restatable}{lemma}{lemfoldingback} 
\label{lem:folding-back}
If there exists an $\ell'$-flat $(\ell,\ell')$-model of $\Kk$ modulo $\Ee$ with bounded degree and bag size then there exists
a finite $(\ell,\ell')$-model of $\Kk$ modulo $\Ee$.
\end{restatable}
 
Combining Lemmas~\ref{lem:consistent},~\ref{lem:unravelling-flat}, and~\ref{lem:folding-back}, we get that  there is a finite $(\ell,\ell')$-model of $\Kk$ modulo $\Ee$ iff there is 
a bounded-degree $\ell'$-flat model of $\Kk$ modulo $\Ee$ whose bags are $\ell'$-consistent and have bounded size. As in an $\ell'$-flat model each bag is a level-$\ell'$ interpretation, when we restrict our search to one-bag models the problem is an instance of the $(\ell', \ell')$-model problem. Models consisting of multiple bags can be built coinductively top-down by means of a greatest fixed point algorithm (similar to type elimination), using  the $(\ell', \ell')$-model problem to check if each bag exists.

\begin{restatable}{lemma}{lemalgoflat}
\label{lem:algo-flat}
The $(\ell,\ell')$-model problem for an $\ALC$ KB $\Kk$, a UCRPQ $\Phi$, and an environment $\Ee$ can be solved in time \[2^{O(\|\Kk\|)+ 2^{\poly(\|\Phi\|)}}\] given an oracle for the $(\ell', \ell')$-model problem (with the same UCRPQ and TBox).
\end{restatable}





\section{Looking Forward (and Back)}
\label{sec:conclusions}

This paper provides first positive results on finite entailment of navigational queries over DLs ontologies. The main technical contribution is  an optimal automata-based \textsc{2ExpTime} upper bound for finite entailment of UCRPQs in $\ALC$.

Let us  take a look back at our journey. We devised an expansion of the semiautomaton used to represent UCRPQs to keep track of  its runs  that begin in all possible states, on all infixes of the input word. By making interpretations and CRPQs knowledgeable of the runs of this expansion,  we are able to associate levels to them as dictated by the transitions of the expansion.  To solve the entailment problem, we use a recursive method eliminating the lowest level from the query and from the interpretation, and solving then the simpler problem. In particular, we look at problem of finding $(\ell, \ell')$ models, and solve it by  recursively  increasing $\ell$ and $\ell'$ in an alternating way, until both reach a maximum level:  Section~\ref{ssec:queries} and~\ref{ssec:interpretations} respectively address the increment of the query level and  of the model. We finally showed what to do when $\ell =\ell' = n+1$, which as argued, is enough to solve the original finite entailment problem. 

As for future work, the first immediate step is to extend our method to deal with test atoms of the form $A?$, which are usually available in UCRPQs. For the ontology language, we believe our method can be adapted to allow inverses, nominals or counting. Regarding more expressive query languages, the natural next step is to consider \emph{two-way} CRPQs. Our current approach relies on the fact that information only flows forward, and it is not clear whether it can be adapted to deal with queries that can go back.


\section*{Acknowledgments}

This work was supported by Poland's National Science Centre grant 2018/30/E/ST6/00042. It also benefited from inspiring discussions with Charles Paperman.

\bibliographystyle{kr}
\bibliography{references}
\appendix
\section{Additional Definitions}

\new{Let us fix $\Kk=(\Tt, \Aa)$.} A \emph{homomorphism} from interpretation $\Ii$ to interpretation
$\Jj$, written as $h : \Ii \to \Jj$ is a function $h : \Delta^\Ii \to
\Delta^\Jj$ that preserves roles, concepts, and individual names: that
is, for all $r \in \roles$, $(h(d), h(e)) \in r^\Jj$ whenever $(d,e)
\in r^\Ii$, for all $A \in \concepts$, $h(d) \in A^\Jj$ whenever $d
\in A^\Ii$, and $h(a)= a$ for all \new{$a \in \Ind(\Aa)$}.

\section{Proofs}

\lemabox*
\begin{proof}
Let $\Ii$ be a finite model of $\Kk = (\Tt, \Aa)$ such that $\Ii\not\models\Phi$. We can think of $\Aa$ as an interpretation with domain \new{$\Ind(\Aa)$}. Then, $\Ii$ contains a subinterpretation $\Ii'$ that is an isomorphic copy of $\Aa$, except that the extension of concepts over \new{$\Ind(\Aa)$} is kept as in $\Ii$. Let $\Jj$ be the interpretation obtained by starting from $\Ii'$ and for each \new{$a \in \Ind(\Aa)$}, adding an isomorphic copy $\Ii_a$ of $\Ii$ sharing only $a$ with $\Ii'$.  Clearly, $\Jj$ is a model of $\Kk$ and $\Jj \not\models\Phi$, because $\Ii$ is a homomorphic image of $\Jj$. Note also, that $\Ii_a\models \Tt$ for all \new{$a\in\Ind(\Aa)$}. This shows that it suffices to look for counter models that are unions of a \emph{core} interpretation $\Jj'$ that is a copy of $\Aa$ up to the interpretations of concept names, and a collection of disjoint \emph{peripheric} models $\Jj_a$ of $\Tt$ such that $\Delta^{\Ii_a} \cap \Delta^{\Jj'} = \{a\}$ for \new{$a \in \Ind(\Aa)$}.

The algorithm iterates through all possible core models $\Jj'$. For each $\Jj'$ it needs to decide if there exist peripheric models $\Jj_a$ for \new{$a\in\Ind(\Aa)$} such that no partial match $\pi$ of $\varphi \in \Phi$ in $\Jj'$ can be extended to a full match of $\varphi$ in the whole $\Jj$. For this it is enough to know if $(\Tt,\Aa')\fentails \Phi'$ where $\Aa'$ ranges over trivial ABoxes using a fixed individual $a$ and concept names from $\CN(\Kk)$, and $\Phi'$ ranges over sets of CRPQs $\varphi_U$ for $\varphi\in\Phi$ and $U \subseteq \var(\varphi)$ defined as follows. The CRPQ $\varphi_U$ is obtained from $\varphi$  by 
\begin{itemize}
    \item dropping all atoms that involve no variable from $U$, as well as all edge atoms involving a variable not in $U$;
    \item replacing each $\Bb_{q,q'}(x,x')$ such that  $x\in U$ and $x\notin U$ with $\Bb_{q,p}(x,a)$ for some $p'$, and each $\Bb_{q,q'}(x,x')$ such that $x\notin U$ and $x' \in U$ with $\Bb_{p,q'}(a,x')$ for some $p$.
\end{itemize}  
Note that $|\Phi'| = 2^{\poly(\|\Phi\|)}$ but all CRPQs in $\Phi'$ have size bounded by $\max_{\varphi \in\Phi} |\varphi|$ and the underlying semiautomaton $\Bb$ is not altered.

The number of possible choices of $\Aa'$ and $\Phi'$ is \[2^{\CN(\Kk)} \cdot 2^{2^{\poly(\|\Phi\|)}}\,.\] The number of distinguishable choices for each peripheric model $\Ii_a$ is 
\[2^{2^{\poly(\|\Phi\|)}}\,.\]
This gives up to \[2^{|\CN(\Kk)|\cdot|\new{\Ind(\Aa)}|}\cdot 2^{2^{\poly(\|\Phi\|)}\cdot|\new{\Ind(\Aa)}|} =  2^{\poly(\|\Kk\|)\cdot 2^{\poly(\|\Phi\|)}}\]
choices for the algorithm. For each choice there are $|\Phi|\cdot|\new{\Ind(\Aa)}|^{O(m)}$ partial matches to consider. The cost of verifying a single match is polynomial in the size of $\Jj'$ and the size of a single $\Phi'$; that is, $\poly(\|\Kk\|,2^{\poly(\|\Phi\|)})$. Overall, the complexity of the algorithm is \[2^{2^{\poly(\|\Phi\|)}\cdot\poly(\|\Kk\|)}\,.\qedhere\]
\end{proof}

\lemlevels*
\begin{proof}
($\Rightarrow$) The run of $\Bb$ from $q$ to $q'$ induces a thread in the run of $\widehat\Bb$. We can split the thread into segments that stay at the same level, giving levels $\ell_1, \ell_2, \dots, \ell_k$, separated by transitions that decrease the level.  Clearly, $1 \leq k \leq n$. The last positions on the subsequent levels give $j_1, j_2, \dots, j_k$. It is easy to check that the corresponding states satisfy the conditions specified in the lemma. 

($\Leftarrow$) The first and third condition, combined with the fact that threads are  non-increasing, imply that between indexes $j_i+1$ and $j_{i+1}$ the thread---from $\delta(q_i, w[j_i+1])$ in $\pp_{j_i+1}$ to $q_{i+1}$ (or $q'$ for $i=m-1$) in $\pp_{j_{i+1}}$---stays at the same level; similarly for the prefix and suffix. Combined with the transitions mentioned in the second condition they give a single thread witnessing a run of $\Bb$ from $q$ to $q'$ on $w$.
\end{proof}

\lemaxioma*
\begin{proof}
It is straightforward to express the condition that no element has incoming edges over different roles from $\Rol(\Kk)$. Pick a fresh concept name  $A_r$ for each $r\in\Rol(\Kk)$ and include axioms $\top \sqsubseteq \forall r.A_r$ and $A_r \sqcap A_s \sqsubseteq \bot$ for all $r,s\in\Rol(\Kk)$ with $r\neq s$.

We next provide an alternative axiomatization of 
\begin{eqnarray}
C_\pp &\sqsubseteq &\forall r. C_{\hat \delta(\pp,r)} \label{eq:trans}\\    
C_\pp \sqcap C_{\pp'} &\sqsubseteq &\bot \label{eq:consistent} \\
 \top &\sqsubseteq  &\bigsqcup_{\pp\in \widehat Q} C_{\pp} \label{eq:states}
\end{eqnarray}
 

%
%
To encode conditions \eqref{eq:states} we include the following axioms for every  $\ell \in\{1,\dots,n\}$.
\[\top \sqsubseteq \bigsqcup_{q\in Q} C_{q,\ell}\]
and 
\[C_{q,\ell} \sqcap C_{q',\ell} \sqsubseteq \bot\]
with  $q,q' \in Q$ such that $q\neq q'$.  
These, together with the following will enforce condition \eqref{eq:consistent}.
\[C_{q,\ell} \sqcap C_{q,\ell'} \sqsubseteq \bot\]
for each pair $\ell,\ell'$ with $\ell \neq \ell'$ and each $q \in Q$.  
%

%

To ensure that the transitions of $\widehat \Bb$ are faithfully represented, we will use auxiliary concepts $A^r_{i,j}, D^r_{i,j}, B^r_i$ with $i,j\in\{1, \dots ,n\}$, and $r$ a role name. 

Let $\mathbf{p}=(p_1, \dots, p_n) \in \widehat Q$, and let $\delta(\mathbf{p},r)= (p'_1, \dots, p'_n)$, for some arbitrary (but fixed) role name $r$.
We will use concept $A^r_{i,j}$ is to indicate that $\delta(q
p_i,r) = p'_k $  for some $k \in \{1,2,\dots, i\}$. Further,  $D^r_{i,j}$ will indicate that $\delta(p_i,r)= p'_j$.
Finally, $B^r_\ell$ will help to indicate that the level of $\mathbf{p}'$ is equal to $\ell$. More precisely, if an element $d$ in the domain encodes the state $\mathbf{p}$, and the level of $\delta(\mathbf{p},r)$ is $\ell$, then every $r$-successor of $d$  must satisfy $B^r_\ell$ (see \eqref{eq:levels}) below.

 We have the following axioms:
\begin{align}
\top &\sqsubseteq  A^r_{1,1} \label{eq:ini} 
%
\end{align}
For every triple $q,\ell,\ell'$ with $q\in Q$ and $\ell,\ell' \in \{1,\dots, n\}$
\begin{align}
C_{q,\ell} \sqcap A^r_{\ell,\ell'} &\sqsubseteq \forall r. (\bigsqcup_{1\leq k \leq \ell'}C_{\delta(q,r),k}) \label{eq:t1} \\
C_{q,\ell} \sqcap \exists r.C_{\delta(q,r),\ell'} &\sqsubseteq D^r_{\ell,\ell'} \label{eq:rec1} 
\end{align}
For every  $\ell,k$ with $1\leq \ell, k < n$:
\begin{align}
D^r_{\ell,k} \sqcap A^r_{\ell,k} &\sqsubseteq A^r_{\ell+1,k+1} \label{eq:updt1} \\
D^r_{\ell,k'} \sqcap A^r_{\ell, k} &\sqsubseteq A_{\ell+1,k} \label{eq:updt12} \quad \text{ for every } k' < k
\end{align}
For every $k \in \{1, \dots, n\}$,  and every $\ell < k$,
\begin{align}
 D^r_{n,\ell} \sqcap A^r_{n,k} &\sqsubseteq  \forall r. B^r_k \label{eq:levels}
    \end{align}
    And for every $\ell < n$,
    \begin{align}
     B^r_\ell &\sqsubseteq (\bigsqcup_{i<j} (C_{q_i,\ell} \sqcap C_{q_j,\ell+1}))    \label{eq:fill}
    \end{align}
Finally, we require that for every $\ell,k,k' \in \{1, \dots , n\}$ such that $k\neq k'$:
\begin{align}
A^r_{\ell,k} \sqcap A^r_{\ell,k'} &\sqsubseteq \bot \label{eq:ic1}\\
D^r_{\ell,k} \sqcap D^r_{\ell,k'} &\sqsubseteq \bot \\
B^r_{k} \sqcap B^r_{k'} & \sqsubseteq \bot \label{eq:ic2}
\end{align}
Intuitively, the axioms encode the listing order of $(\delta(p_1), \dots, \delta(p_n))$ in $\hat \delta(\mathbf{p})= (p'_1, \dots, p'_n)$ as follows.
\eqref{eq:ini} encodes that the first position of the tuple is (the only)  available for $\delta(p_1,r)$. Further, by \eqref{eq:t1} we have that if $q=p_\ell$ and the next available position for  $p_\ell$ is $\ell'$, then $\delta(p_\ell,r) = p'_{k}$ for some $1\leq k \leq \ell'$, which means in particular that 
$\delta(p_1,r) = p'_1$. 

As mentioned above, \eqref{eq:rec1} is used for ``recording" the level of $\delta(p_\ell,r)$ using the concept $D^r_{\ell, \ell'}$. That is, $D^r_{\ell, \ell'}$ holds whenever $\delta(p_\ell,r)= p'_{\ell'}$. 

Clearly, if $\delta(p_\ell,r)= p'_k$, and its next available position was $k$, then the next available position for $\delta(p_{\ell+1})$ is $k+1$. This situation is captured by  \eqref{eq:updt1}. On the other hand, if $\delta(p_\ell,r)$ does not takes position $k$ (which is only possible if $\delta(p_\ell,r) = p'_{k'}$, with $k' < k$) then $k$ is available for $\delta(p_{\ell+1},r)$, as captured by \eqref{eq:updt12}.

Now, we need to account for the positions not taken by any $\delta(p_i,r)$. By the way the positions are taken, it is enough to record the smallest position unused. This information is encoded using the concept $B^r_k$.  
 Thus, if $A^r_{n,k}$ and $D_{n,\ell}$, for some $\ell < k$, are both satisfied then $k$ is the next available position for listing the remaining states ordered as in $Q= q_1, \dots q_n$. This is captured by~\eqref{eq:levels} and \eqref{eq:fill}.  
 
Finally, the role of CIs \eqref{eq:ic1}--\eqref{eq:ic2} is to ensure the consistency of the information encoded.

With this intuition in mind, it is not difficult to see~\eqref{eq:trans} is faithfully encoded. 
\end{proof}

\lemreach*
\begin{proof}
($\Rightarrow$) This is obvious because the definition of $\Bb^\Ii_{q,q'}$ requires a path from $e$ to $e'$ in $\Ii$.

($\Leftarrow$)
Consider a path from $e$ to $e'$ in $\Ii$ and the corresponding run of the automaton $\widehat\Bb$.
We will focus on the thread starting in the state $q$.
It starts on level $\ell$ because $e \in C^\Ii_{q,\ell}$ and cannot drop below level $\ell$ because $\Ii$ is a level-$\ell$ interpretation. Hence, the thread ends in $e'$ on level $\ell$. But the state in $e'$ on level $\ell$ is $q'$.
Thus, this thread corresponds to a correct run of $\Bb$ from $q$ to $q'$.
\end{proof}

\lemdecoratemodel*
\begin{proof}

$\Ii \times \widehat\Bb$ is a $\widehat\Bb$-decorated interpretation by construction. 

The mapping  $(e, r, \pp) \mapsto e$ from $\Ii \times \Bb$ to $\Ii$ is a homomorphism. Should some $\phi\in\Phi$ be matched in $\Ii \times \widehat\Bb$, one could compose the match with the homomorphism above to obtain a match in $\Ii$.

Assume that $\Ii \models \Kk$. Satisfaction of the ABox transfers directly to $\Ii \times \Bb$. Let us see that  $\Ii \times \Bb$ is a model of the TBox of $\Kk$. For CIs of the forms $\bigcap_i A_i \sqsubseteq \bigsqcup_j B_j$ and $A \sqsubseteq \forall r.B$ this follows from the existance of a homomorphic mapping from $\Ii \times \Bb$ to $\Ii$, described above. 
For CIs of the form $A \sqsubseteq \forall r.B$ the reason is that the transition function of $\widehat \Bb$ is defined for each state $\pp$ of $\widehat\Bb$ and each $r\in\Rol(\Kk)$, thus each $r$-edge originating in $e$ will have its counterpart originating in $(e,s,\pp)$ for each $s$ and $\pp$.
\end{proof}

\lemconsistent*
\begin{proof}
By contradiction, suppose that $\Ii\models\Phi$. Then there exists a match $\pi$ of some $\phi \in \tilde \Phi$ in $\Ii$. Take any element $e$ in the image of $\pi$. The definition of consistency applied for the trivial partition  with $\varphi = \varphi'$  
implies that $e \in A^\Ii_{\varphi,V}$ for $V=\pi^{-1}(a)$.
\end{proof}

\lemconsistentcomp*
\begin{proof}

Left to right implication is obvious. For right to left implication assume the contrary, that all bags and edge-bags are $\ell$-consistent, but there is:
\begin{itemize}
    \item a fragment $\varphi$ of $\Phi$,
    \item a partition of $\varphi$ into a CRPQ $\varphi'$ of level $\ell$ and fragments $\varphi_1, \varphi_2, \dots, \varphi_k$, and sets $V$, $V_1, V_2, \dots, V_k$ such that:
    \begin{itemize}
        \item  $\var(\varphi_i)\cap \var(\varphi_j) = \emptyset$ for $i\neq j$,
        \item $V_i = \var(\varphi_i) \cap \var(\varphi')$,
        \item $\emptyset \neq V\subseteq \var(\varphi) \cap \var(\varphi')$;
    \end{itemize}
    \item a match $\pi$ for $\varphi'$ in $\Ii$ and functions $\kappa$, $\kappa_1, \dots, \kappa_k$ such that:
    \begin{itemize}
        \item $\pi(V_i) = \{e_i\} \subseteq \big(A_{\varphi_i, V_i}^{\kappa_i}\big)^{\Ii}$ for all $i$,
        \item $\pi(V) = \{e\} \not\subseteq (A_{\varphi, V}^\kappa)^{\Ii}$,
        \item $\kappa_i(x) \leq \ell$ for all $x \in  \var(\varphi_i)\,$,
        \item $\kappa(x) = \kappa_i(x)$ for all $x \in  \var(\varphi_i) \setminus V_i\,$,
        \item $\kappa (x) = \ell$ for all $x \in \var(\varphi) \cap \var(\varphi')\,$.
    \end{itemize}
\end{itemize}
For each RPQ atom $\Bb_{q,q'}(x, y)$ in $\varphi'$, choose a path from $\pi(x)$ to $\pi(y)$ witnessing that the atom is satisfied.
Pick the parameters above and the witnessing paths for which 
$\pi$ spans through the smallest number of bags and edge-bags (we count a bag if some edge atom in $\varphi'$ is mapped by $\pi$ to an edge of this bag, or if some witnessing path shares an edge with this bag).
Note that the match of $\varphi'$ given by $\pi$ and the witnessing paths is connected. This is because the whole query $\varphi$ is connected and because $\pi(V_i)$ consists of just one element for each $i$ -- the query $\varphi'$ itself might not be connected, although it would be if we equated all variables in each $V_i$.
The number of bags $\pi$ spans through must be at least two: were it contained in one bag, this bag would be inconsistent. Essentially, we will show that we can derive the fact that $e \in (A_{\varphi, V}^\kappa)^{\Ii}$ from $\ell$-consistency conditions for some matches spanning through smaller number of bags.

Let $b$ be the bag of $e$ (not edge-bag, so it is unique). The match $\pi$ necessarily spans through the bag $b$: otherwise no edge or RPQ is matched inside $b$, so $\varphi'$ consists only of unary atoms, which means that $k = 1$, $\varphi_1 = \varphi$ and $e = e_1$, which easily leads to contradiction.
Let $\psi', \psi_1, \dots, \psi_m$ be a partition of $\varphi$ taking into account the bag $b$, match $\pi$ and the chosen witnessing paths, where $\psi_1, \dots, \psi_m$ are fragments. That is, in the definition of a partition:
\begin{itemize}
    \item the initial set $X'$ is the set of variables of $\varphi$ which are mapped by $\pi$ to the bag $b$ (note that $\pi$ is defined only on $\var(\varphi')$);
    \item each RPQ is split (or not) in an appropriate way, depending on whether the corresponding witnessing path has zero, one, or two endpoints in the bag $b$, and the fresh variables are assigned level and state according to the last (or first) elements in $b$ on the corresponding witnessing paths;
    \item the sets $X_i$ are chosen in the way which results in $\psi_i$ being fragments.
\end{itemize}
Let $U_i = \var(\psi_i) \cap \var(\psi')$ and $\pi'$ be a match agreeing with $\pi$ on $\var(\psi')\cap\var(\varphi')$, extended with fresh variables mapped to the appropriate elements on witnessing paths. 
The set $\pi'(U_i)$ consists of exactly one element, call it $e'_i$. Indeed, for some $i$ there is $j$ such that $\psi_i = \varphi_j$ and $U_i = V_j$, and so $\pi'(U_i) = \pi(V_j) = \{e_j\}$; and for other $i$, $\psi_i$ consists of some part of $\varphi'$, possibly merged with some $\varphi_j$, with all the variables shared with $\psi'$ being matched to an endpoint of an edge leaving or entering bag $b$.

We claim that for each $i$, $e'_i \in A^{\lambda_i}_{\psi_i, U_i}$ for some appropriate $\lambda_i$, such that we can use $\ell$-consistency for $b$ to show that $e \in (A_{\varphi, V}^\kappa)^{\Ii}$.

To show that, we need to relate all $\varphi_j$ to some $\psi_i$. Specifically, for each $i \in \{1, \dots, m\}$, consider all $j \in \{1, \dots, k\}$ such that $\var(\psi_i) \cap \var(\varphi_j) \neq \emptyset$. Note that variables in all $\psi', \psi_1, \dots, \psi_m$ are exactly the variables of $\varphi$, along with some fresh variables splitting some RPQs, and analogously for $\varphi', \varphi_1, \dots, \varphi_k$ (and each of the fragments $\varphi_j$ shares at least one variable with $\varphi$). Since $\psi'$ can be seen as a part of $\varphi'$ \big(in particular, $\var(\varphi)\cap\var(\psi') \subseteq \var(\varphi)\cap\var(\varphi')$\big), it is easy to see that each $j$ will be assigned to some $i$. 
For each $i$, $\psi_i$ can be (yet again) partitioned into $\psi'_i$ (intuitively being the common part of $\varphi'$ and $\psi_i$) and $\varphi_j$ for all $j$ assigned to this $i$. There is also a match $\pi_i$ agreeing with $\pi$ on $\var(\psi_i')\cap\var(\varphi')$, as usual, extended with fresh variables from $U_i$ mapped to the appropriate elements on witnessing paths.
Since this match spans through fewer bags than $\pi$ (it does not span through the bag $b$, as all edges and parts of RPQs that were mapped inside $b$ are in $\psi'$), the $\ell$-consistency condition must be satisfied for this choice of parameters; since $\pi(V_j) = \{e_j\} \subseteq \big(A_{\varphi_j, V_j}^{\kappa_j}\big)^{\Ii}$ for all $j$, we get that $e'_i \in A^{\lambda_i}_{\psi_i, U_i}$ for $\lambda_i$ which agrees with $\kappa_j$ on all $\var(\varphi_j) \setminus V_j$, and is equal to $\ell$ on all other variables.
Using this fact for all $i$ and using $\ell$-consistency in the bag $b$ for the match $\pi'$, we get that $e \in (A_{\varphi, V}^\kappa)^{\Ii}$ for $\kappa$ which agrees with each $\kappa_i$ on all variables from $\varphi_i$ apart from $V_i$ and is equal to $\ell$ on all other variables, which is exactly what needed to be shown.
\end{proof}

\lemunravellingconnected*
\begin{proof}
This is proved by routine unravelling. Let $\Ii$ be a finite $(\ell,\ell')$-model of $\Kk$ modulo $\Ee$. For each element $d$ in $\Ii$ define $\Ii_d$ as the subinterpretation of $\Ii$ obtained by restricting the domain of $\Ii$ to the elements in the maximal strongly connected subset of $\Delta^\Ii$ that contains $d$.
\new{Recall that the ABox $\Aa$ of $\Kk$ is trivial. Let $a$ be the unique element of $\Ind(\Aa)$.}  Begin the construction of a tree-like model $\Jj$ by taking a copy of $\Ii_a$. Then, as long as there exists an element $e$ in $\Jj$ and a CI $A \sqsubseteq \exists r. B$ in the TBox of $\Kk$ such that $e\in A^\Jj$ but there is yet no $e'\in B^\Jj$ such that $(e,e')\in r^\Jj$, find the original $d$ of $e$ in $\Ii$ and an element $d'\in B^\Ii$ such that $(d,d')\in r^\Ii$. Add to $\Jj$ a copy of $\Ii_{d'}$ as a new bag, with an $r$-edge from $e$ to the copy of $d'$. This construction gives a finite interpretation: the height of the tree of bags associated to $\Jj$ is bounded by the height of the DAG of strongly connected components of $\Ii$. It is straightforward to check that $\Jj$ is a level-$\ell$ model of $\Kk$ modulo $\Ee$. It is also clear that $\Jj$ can be mapped homomorphically to $\Ii$ by mapping each element of $\Jj$ to its original in $\Ii$. Because $\ell'$-consistency is defined in terms of forbidden matches it follows immediately that $\ell'$-consistency of $\Ii$ implies $\ell'$-consistency of $\Jj$.
\end{proof}

\crpqSCCreducts*
\begin{proof}
Use Lemma \ref{lem:reach} and the definition of reducts.
\end{proof}

\lemenvreduct*
\begin{proof}[Proof]
The proof will use constructions that are very similar to the ones used in the proof of Lemma~\ref{lem:consistentcomp}.
Let us start with right to left implication. Let $\Ii$ be some $(\ell+1)$-consistent model of $\Kk$ modulo $\Ee'$ for some $(\ell+1)$-reduct $\Ee'$ of $\Ee$. We will show that the exact same model $\Ii$ is a strongly $\ell$-consistent model of $\Kk$ modulo $\Ee$ (note that strong $\ell$-consistency does not mention concept names from $\CN^\Phi_{\ell+1}$, so adjusting them is not necessary for this implication). Assume the contrary; so there is:
\begin{itemize}
    \item a fragment $\varphi$,
    \item a partition of $\varphi$ into a CRPQ $\varphi'$ of level $\ell$ and fragments $\varphi_1, \varphi_2, \dots, \varphi_k$, and sets $V$, $V_1, V_2, \dots, V_k$ such that:
    \begin{itemize}
        \item  $\var(\varphi_i)\cap \var(\varphi_j) = \emptyset$ for $i\neq j$,
        \item $V_i = \var(\varphi_i) \cap \var(\varphi')$,
        \item $\emptyset \neq V\subseteq \var(\varphi) \cap \var(\varphi')$;
    \end{itemize}
    \item a match $\pi$ for some $(\ell+1)$-reduct $\psi'$ of $\varphi'$ in $\Ii$ and functions $\kappa$, $\kappa_1, \dots, \kappa_k$ such that:
    \begin{itemize}
        \item $\pi(V_i) = \{e_i\} \subseteq \big(A_{\varphi_i, V_i}^{\kappa_i}\big)^{\Ii}$ for all $i$,
        \item $\pi(V) = \{e\} \not\subseteq (A_{\varphi, V}^\kappa)^{\Ii}$,
        \item $\kappa_i(x) \leq \ell$ for all $x \in  \var(\varphi_i)\,$,
        \item $\kappa(x) = \kappa_i(x)$ for all $x \in  \var(\varphi_i) \setminus V_i\,$,
        \item $\kappa (x) = \ell$ for all $x \in \var(\varphi) \cap \var(\varphi')\,$.
    \end{itemize}
\end{itemize}
For each RPQ atom $\Bb_{q,q'}(x, y)$ in $\psi'$, choose a path from $\pi(x)$ to $\pi(y)$ witnessing that the atom is satisfied.

When constructing $\Ee'$, the exact same parameters were considered ($\varphi$, its partition into $\varphi'$, $\varphi_1, \varphi_2, \dots, \varphi_k$, $(\ell+1)$-reduct $\psi'$ of $\varphi'$, the set $V$ and the function $\kappa$), sets $U_1, \dots, U_m$ and fragments $\psi_1, \dots, \psi_m$ were defined (note that for a fixed match $\pi$ and witnessing paths, partial matches and partial witnessing paths for $\psi_i$ can also be obtained, for the parts common with $\psi'$), and a choice was made:
\begin{itemize}
\item either pick $i$ such that $U_i = \emptyset$ and remove  all unary types that contain any $A^{\lambda_i}_{\psi_i,W_i}$ with $W_i\subseteq \var(\psi_i) \cap \var(\psi')$, 
$\lambda_i\big(\var(\psi_i) \cap \var(\psi')\big)=\{\ell+1\}$, and 
$\lambda_i(x)=\kappa(x)$ for all $x \in \var(\psi_i)\setminus \var(\psi')$ ;
\item  or remove all unary types 
that contain some $A^{\lambda_i}_{\psi_i,U_i}$ with
$\lambda_i\big(\var(\psi_i) \cap \var(\psi')\big)=\{\ell+1\}$ and $\kappa(x)=\lambda_i(x)$ for all $x \in \var(\psi_i)\setminus \var(\psi')$,
for each $i$ such that $U_i \neq \emptyset$, but do not contain $A^\kappa_{\varphi,V}$.
\end{itemize}
For each $i \in \{1, \dots, m\}$ consider a partition of $\psi_i$ into $\psi'_i$ (intuitively being the common part of $\psi_i$ and $\psi'$) and fragments $\varphi_j$ for all $j$ such that $\var(\varphi_j)\cap\var(\psi_i) \neq \emptyset$, and a match $\pi_i$ of $\psi'_i$ agreeing with $\pi$ and the choice of witnessing paths.

If the first choice was made, take the chosen $i$, choose any element $e'$ in the image of $\var(\psi_i)\cap\var(\psi')$ under $\pi_i$, and let $W_i$ be the preimage of $e'$ under $\pi_i$. By $(\ell+1)$-consistency, we see that $e' \in A^{\lambda_i}_{\psi_i, W_i}$, where $\lambda_i$ agrees with $\kappa_j$ (and $\kappa$) on $\var(\varphi_j)\setminus V_j$ for all $j$ such that $\var(\varphi_j)\cap\var(\psi_i) \neq \emptyset$, and is equal to $\ell+1$ on all other variables. However, all unary types containing this concept were forbidden in $\Ee'$, which is a contradiction with $\Ii$ being a model of $\Kk$ modulo $\Ee'$.

If the second choice was made, use $(\ell+1)$-consistency for each $\psi_i$ with $U_i \neq \emptyset$ and its partition as described above, which proves that $e \in A^{\lambda_i}_{\psi_i, U_i}$ for some $\lambda_i$ agreeing with $\kappa_j$ (and $\kappa$) on $\var(\varphi_j)\setminus V_j$ for all $j$ such that $\var(\varphi_j)\cap\var(\psi_i) \neq \emptyset$, and is equal to $\ell+1$ on all other variables. Thus, because of the environment $\Ee'$, $A^\kappa_{\varphi, V}$ must also be satisfied. (Note that the union of $\var(\psi_i) \setminus \var(\psi')$ for all $i$ with $U_i \neq \emptyset$ is equal to the union of $\var(\varphi_j) \setminus V_j$ for all $\varphi_j$ which are not disjoint with all $\psi_i$ considered here).

Now we will prove left to right implication. Assume that $\Ii$ is a strongly $\ell$-consistent model of $\Kk$ modulo $\Ee$.
Let $\Ii'$ be an interpretation that agrees with $\Ii$ over all role and concept names except $\CN_{\ell+1}^\Phi$, and in which the interpretation of these concept names is \emph{correct} in the following sense. First, for each $e\in \left( A^{\kappa}_{\varphi,V}\right)^{\Ii'}$ with $\CN_{\ell}^\Phi$  we let $e\in \left( A^{\kappa'}_{\varphi,V}\right)^{\Ii'}$ where $\kappa'(x)=\ell+1$ for $x\in V$ and $\kappa'(x)=\kappa(x)$ for $x \in \var(\varphi)\setminus V$. Then, we add element $e$ to $(A^\kappa_{\varphi, V})^{\Ii'}$ with $(A^\kappa_{\varphi, V})^{\Ii'}\in\CN^{\Phi}_{\ell+1}$ if and only if there exists some partition of $\varphi$ into $\varphi', \varphi_1, \dots, \varphi_k$, sets $V_1, \dots, V_k$ and a match $\pi$ for $\varphi'$ in $\Ii'$ with all requirements exactly as in the definition of $(\ell+1)$-consistency, in which additionally $\kappa_i(x) \leq \ell$ for all $x \in \var(\varphi_i)\setminus V_i$ and $\kappa_i(x) = \ell+1$ for all $x \in V_i$, for $i \in \{1, \dots, k\}$.
Note that this does not leave many choices regarding the partition: all variables $x$ such that $\kappa(x) = \ell+1$ must be in $\varphi'$, all other variables of $\varphi$ must be outside $\varphi'$, so it is even known which RPQ atoms are split; only the levels and states of the fresh variables might differ between different partitions.

The interpretation $\Ii'$ is $(\ell+1)$-consistent, since any partition of some $\varphi$ and a corresponding match (as in the definition of $(\ell+1)$-consistency) for which $\kappa_i(V_i) = \{\ell+1\}$ for some $i$, can be ``unwrapped'' to ones satisfying the correctness condition, as follows. If $e_i \in A^{\kappa_i}_{\varphi_i, V_i}$ and $\kappa_i(V_i) = \{\ell+1\}$, by the correctness condition, there is a partition of $\varphi_i$ and an appropriate match witnessing that. This partition and match of $\varphi_i$ can me ``merged'' into the partition and match of $\varphi$. Applying this procedure for all $i$ such that $\kappa_i(V_i) = \{\ell+1\}$  results in a partition and a match of $\varphi$ to which the correctness condition can be applied.

Now we will construct an appropriate $(\ell+1)$-reduct $\Ee'$ of $\Ee$. Consider some parameters $\kappa$, $\varphi, \varphi', \varphi_1, \dots, \varphi_k, \psi'$, $V$, define $\psi_1, \dots, \psi_m, U_1, \dots, U_m$ as when constructing an $(\ell+1)$-reduct of the environment; recall that $\kappa$ uses levels at most $\ell$, and $\kappa\big(\var(\varphi)\cap\var(\varphi')\big) = \{\ell\}$.
We need to choose one of the two options mentioned in the construction. If for some $i$ such that $U_i = \emptyset$, the interpretation $\Ii'$ does not violate the restrictions imposed by the first choice (that is,  there are no elements in concept $A^{\kappa_i}_{\psi_i,W_i}$ in $\Ii'$ for all $W_i$ and $\kappa_i$ as defined during the construction of a reduct), choose this option and this $i$. Otherwise, choose the second option, knowing that in $\Ii'$ for each $i$ such that $U_i = \emptyset$ there is some element $e_i \in A^{\kappa_i}_{\psi_i,W_i}$ for some $W_i$ and $\kappa_i$ such that $W_i\subseteq \var(\psi_i) \cap \var(\psi')$, 
$\kappa_i\big(\var(\psi_i) \cap \var(\psi')\big)=\{\ell+1\}$, and 
$\kappa_i(x)=\kappa(x)$ for all $x \in \var(\psi_i)\setminus \var(\psi')$.


We claim that $\Ii'$ is a model of $\Kk$ modulo the environment $\Ee'$ constructed as above.
Suppose this is not the case. By the construction of $\Ee'$, the interpretation $\Ii'$ satisfies all the requirements imposed by the first choice. Hence, it must be the case that $\Ii'$ does not satisfy the requirements imposed by the second choice; that is, in $\Ii'$ there is an element $e$ in  concepts $A^{\kappa_i}_{\psi_i,U_i}$ where
$\kappa_i\big(\var(\psi_i) \cap \var(\psi')\big)=\{\ell+1\}$ and $\kappa(x)=\kappa_i(x)$ for all $x \in \var(\psi_i)\setminus \var(\psi')$,
for all $i$ such that $U_i \neq \emptyset$,
 but not in $A^\kappa_{\varphi,V}$.

For each $i \in \{1, \dots, m\}$ and the witnessing element $e_i \in \left(A^{\kappa_i}_{\psi_i,W_i}\right)^{\Ii'}$ (if $U_i \neq \emptyset$, we let $e_i = e$ and $W_i = U_i$), take the partition and match $\pi_i$ guaranteed by the correctness condition; all variables from $\var(\psi_i) \cap \var(\psi')$ are guaranteed to be in the domain of $\pi_i$.
Merge all these matches together into one match $\pi'$. It is nearly a match of $\psi'$; since $\kappa_i$ determines for each $i$ which RPQ atoms in $\psi_i$ will be split in the partition, and it agrees with $\kappa$ on this matter, we can identify the splitting variables in $\psi'$ and in the domain of $\pi_i$; but these variables \big($\var(\psi')\setminus \bigcup_i\var(\psi_i)$\big) might have different states and levels assigned in $\pi'$; since $\Ii$ is a level-$\ell$ interpretation, the assigned levels are at most $\ell$.
Let $\tilde\psi'$ be an $(\ell+1)$-reduct of $\psi'$ with states and levels of these variables adjusted to match the ones in $\pi'$.
Let $\tilde\pi$ be $\pi'$ adjusted to $\tilde\psi'$ (it is easy to modify a match for a CRPQ to a match of its $(\ell+1)$-reduct). We get that there is a partition of $\varphi$ into $\tilde\varphi'$, $\tilde\varphi_1, \dots, \tilde\varphi_h$, sets $\tilde{V_1}, \dots, \tilde{V_h}$ and the match $\tilde\pi$ of an $(\ell+1)$-reduct $\tilde\psi'$ of $\tilde\varphi'$ to $\Ii'$, where $\tilde\pi(\tilde{V_i}) = \{\tilde{e_i}\}$, $\tilde{e_i} \in \left(A^{\lambda_i}_{\tilde{\varphi_i}, \tilde{V_i}}\right)^{\Ii'}$, and $\lambda_i$ do not use level $\ell+1$.

Using strong $\ell$-consistency of $\Ii$ with these parameters gives us that $e \in \left(A^\lambda_{\varphi,V}\right)^\Ii$ for some $\lambda$. One just needs to show that $\lambda(x) = \kappa(x)$ for all $x \in \var(\varphi)$ to arrive at a contradiction.
The matches $\pi_i$ obtained from correctness condition for $e_i \in \left(A^{\kappa_i}_{\psi_i, W_i}\right)^{\Ii'}$ guarantee that $\kappa_i(x) = \kappa(x)$ for all $x \in \var(\psi_i) \setminus \var(\psi')$. We also know that all $x \in \var(\psi_i) \cap \var(\psi')$ are in the domain of $\pi_i$, so (by the above use of strong $\ell$-consistency) $\lambda(x) = \ell$ for $x \in \var(\psi_i) \cap \var(\psi')$, and $\lambda(x) = \kappa(x)$ for all other variables of $\varphi$, which is exactly what is needed.
%
%
%
%
\end{proof}

\lemalgoconnected*
\begin{proof}
By Lemma~\ref{lem:strongly-consistent} it suffices to decide if there is a finite tree-like level-$\ell$ model of $\Kk=(\Tt,\Aa)$ modulo $\Ee=(\Theta,\varepsilon)$ whose edge-bags are $\ell$-consistent and bags are strongly $\ell$-consistent.  Our algorithm will compute the set of unary types that are realizable in such  interpretations of increasing height. Here, by the height of a tree-like interpretation we mean the number of edges on the longest path from the root bag to a leaf bag.

Th algorithm begins from the empty set of types $\Phi_{0} = \emptyset$. In round $h=1, 2, \dots$, based on the set $\Phi_{h-1}$ it computes the set  $\Phi_{h}$ of types that can be realized in models of height $h-1$. The type $\tau$ is added to $\Phi_h$ iff there exists a finite strongly $\ell$-consistent level-$\ell$ model of $\Kk_\tau=(\Tt,\Aa_\tau)$ modulo $\Ee_h$ where \[\Aa_\tau = \left\{A(a) \bigm| A\in\tau\right\}\] for a designated $a\in\mn{N_I}$ and $\Ee_h$ is defined based on $\Ee$ and $\Phi_{h-1}$ as explained below; for the existence test we use Lemma~\ref{lem:envreduct}.

For unary types $\tau_1, \tau_2$ and a role name $r$ let $\Jj_{(\tau_1,r,\tau_2)}$ be the edge-bag built from an element $e_1$ of type $\tau_1$ and an element $e_2$ of type $\tau_2$ connected by an $r$-edge.
We let $\Ee_h=(\Theta,\varepsilon_h)$ and include $(r,B)$ in $\varepsilon_h(\tau_1)$ for $\tau_1\in\Theta$ iff either $(r,B)\in\varepsilon(\tau_1)$ or there exists $\tau_2\in\Theta_{h-1}$ such that $B\in\tau_2$ and $\Ii_{(\tau_1,r,\tau_2)}$ is an $(\ell,\ell)$-interpretation
and satisfies all CIs of the form $A'\sqsubseteq \forall r.B'$ in $\Tt$.
Recall that the first condition amounts to checking that $\Ii_{(\tau_1,r,\tau_2)}$ is $\widehat\Bb$-decorated  ($\Ii_{(\tau_1,r,\tau_2)}\models\widehat\Tt_\Bb$), level-$\ell$, and $\ell$-consistent.

Because the computed sets satisfy
\[\Phi_0 \subseteq \Phi_1 \subseteq \dots \subseteq \Phi_{h-1} \subseteq \Phi_{h} \subseteq \dots ,\] after at most $2^{|\CN(\Kk)|+2^{\poly(\|\Kk\|)}}$ rounds the sets $\Phi_h$ stabilize. The algorithm should return yes iff the last $\Phi_h$ contains a unary type compatible with the assertions made by the Abox $\Aa$ of $\Kk$ about the unique individual it mentions. 
It is not hard to check that each round can be performed in time $2^{O(\|\Kk\|)+2^{\poly(\|\Kk\|)}}$, yielding the desired complexity upper bound. 
\end{proof}

\lemaldiscrete*

\begin{proof}
Let us first see when a discrete interpretation is $(n+1)$-consistent. Consider a partition of a fragment $\varphi$ into $\varphi', \varphi_1,\dots, \varphi_k$ like in the definition of  $(n+1)$-consistency. Because $\varphi'$ has level $n+1$, it must be a UCQ. If $\varphi'$ contains a binary atom, it cannot be matched in a discrete interpretation.  Hence, we can assume that $\varphi'$ contains no binary atoms. Fragments $\varphi_1,\dots, \varphi_k$  share no variables, but $\varphi$ is connected, so $k=1$. It follows that $\varphi_1 = \varphi$ and $V \subseteq \var(\varphi') \subseteq V_1$. Then,  $(n+1)$-consistency reduces to the condition $A^\kappa_{\varphi,V_1} \sqsubseteq A^\kappa_{\varphi,V}$ for all $\kappa$ such that $\kappa(x) = n+1$ for all $x \in V_1$. We can capture $(n+1)$-consistency of the model by replacing $\Ee$ with the environment $\Ee'=(\Theta',\epsilon')$ obtained from $\Ee$ by filtering out unary types that violate this condition. This can be done in time polynomial in the size of $\Ee$.

It remains to decide if there is a discrete model of $\Kk$ modulo $\Ee'$. 
This is the case iff
for the individual $a$ mentioned in $\Aa$ there is a type $\tau \in \Theta'$ compatible with the assertions on $a$ in $\Aa$ such that for each concept inclusion $A \sqsubseteq \exists r. B$ in $\Tt$ if $A \in \tau$ then $(r,B) \in \varepsilon'(\tau)$. This can be checked in time polynomial in the size of $\Kk$ and $\Ee'$.

Overall,  the existence of an $(n+1,n+1)$-model of $\Kk$ modulo $\Ee$ can be decided in time polynomial in the size of $\Kk$ and $\Ee$; that is, in time $2^{O(\|\Kk\|)+2^{\poly(\|\Phi\|)}}$.
\end{proof}

\lemunravellingflat*
\begin{proof}
This is also proved by routine unravelling, much like Lemma~\ref{lem:unravelling-connected}. The difference is that this time for $\Ii_d$ we take the interpretation $\Ii$ with all edges of level strictly below $\ell'$ removed. This unravelling procedure may pass through the same element multiple times on the same branch, so the resulting tree-like structure $\Jj$ may be infinite.  Because new bags are added to $\Jj$ only when a witness is missing in the parent bag, it follows that all edges between bags have level strictly below $\ell'$ (all edges of level at least $\ell'$ are already copied in the parent bag, together with their targets). Hence, $\Jj$ is $\ell'$-flat. The size of each bag is equal to the size of $\Ii$ and the degree within each bag is bounded by the maximal degree in $\Ii$. The number of child bags connected to the same element in the parent bag is bounded by the size of the TBox. Hence, the degree in $\Jj$ is bounded.  Checking that $\Jj$ is a $(\ell,\ell')$-model of $\Kk$ modulo $\Ee$ is straightforward, just like in Lemma~\ref{lem:unravelling-connected}.
\end{proof}

\lemboundedpaths*
\begin{proof}
Suppose a $\widehat\Bb$-decorated CRPQ $\varphi$ is matched in a $\widehat\Bb$-decorated interpretation $\Jj$. Each path in $\Jj$ corresponds to the run $\rho$ of $\widehat B$ obtained by reading the states decorating the elements on the path. Such a path witnesses an RPQ atom $\Bb_{q,q'}(x,x')$ iff the thread of $\rho$ beginning in $q$ ends in $q'$.  If the atom has end level at least $\ell'$, then the level of $q'$ in the last state of $\rho$ must be at least $\ell'$.  Observe however that each edge of level strictly below $\ell'$ brings all threads from levels $\ell'$ and higher at least one level down. Consequently, the path may use at most $n - \ell'$ edges of level strictly below $\ell'$.
\end{proof}

\lemboundedcrpq*
\begin{proof}
Consider a level-$\ell'$ CRPQ $\varphi$ matched in $\Jj$. Consider a witnessing path $e_0, e_1, \dots, e_k$ in $\Jj$. Let $\pp_0, \pp_1, \dots, \pp_k$ be the run of $\widehat \Bb$ corresponding to the witnessing path and let $\ell_0, \ell_1, \dots, \ell_k$ be the thread in $\rho$ that corresponds to the witnessing run $q_0,q_1,\dots, q_k$ of $\Bb$. We have $\ell_0 \geq \ell_1 \geq \dots \geq \ell_k \geq \ell'$. Suppose that for some $i < j$ we have $e_i = e_j$ and $\ell_i = \ell=j$. It follows immediately that $\pp_i = \pp_j$ and $q_i = q_j$. Thus, we can choose a shorter witnessing path by skipping $e_{i+1}, e_{i+2}, \dots, e_{j}$. Consequently, it is enough to look at witnessing paths that visit each element at most $(n-\ell'+1)$ times. From the assumption on the structure of $\Jj$ and from Lemma~\ref{lem:bounded-paths} it follows that every simple witnessing path has length at most $M(n-\ell'+1)^2$. 
\end{proof}

\lemalgoflat*
\begin{proof}

By Lemma~\ref{lem:unravelling-flat} it suffices to decide if there exists an $\ell'$-flat $(\ell,\ell')$-model of $\Kk$ modulo $\Ee$ with bounded degree and bag size. Using Lemma~\ref{lem:consistentcomp} and the definition of $\ell'$-flatness this is amounts to deciding if there exists a (possibly infinite) tree-like model of $\Kk$ modulo $\Ee$ such that each bag is a finite $(\ell',\ell')$-interpretation, each   edge-bag is an $(\ell,\ell')$-interpretation but not level-$\ell'$, and the size of bags and the degree of elements is bounded.

The algorithms is similar to the one in Lemma~\ref{lem:algo-connected}. The main difference is that the model can now be infinite. Suppose for a while, however, that we are interested in computing only finite models. Then we can proceed just like in Lemma~\ref{lem:algo-connected}, computing sets \[\emptyset= \Phi_0 \subseteq \Phi_1 \subseteq \Phi_2 \subseteq \dots\] but as we are after $\ell'$-consistent bags, rather than strongly $\ell'$-consistent, we reduce directly to the $(\ell',\ell')$-model problem for $\Kk_\tau$ defined like before, and $\Ee_h$ define almost like before, the difference being that we additionally require that the edge in $\Jj_{(\tau_1,r,\tau_2)}$ has level strictly below $\ell'$. We do not need to do anything about the size of the bags and the degree, because in a finite interpretation these are always bounded. 

In order to take into account also infinite models, we replace induction by co-induction.  The algorithm proceeds just like described above but it starts from $\Phi_0$ containing all unary types over $\CN(\Kk)$ and concepts $C_{q,k}$ and $A_{\psi,V}^\kappa$. It follows that \[\Phi_0 \supseteq \Phi_1 \supseteq \dots \supseteq \Phi_{h-1} \supseteq \Phi_h \supseteq \dots\,.\]
Like before, the sequence must stabilize after at most $2^{|\CN(\Kk)|+2^{\poly(\|\Kk\|)}}$ steps. We claim that the algorithm can answers yes iff the last computed $\Phi_h$ contains a type compatible with the ABox of $\Kk$. This is because one can built the potentially infinite model top down plugging in as bags the models witnessing the addition of $\tau$ to $\Phi_h$ in the last round. Importantly, the number of these witnesses is finite, because the number of invoked instances of the $(\ell',\ell')$-model problem is finite. Consequently, the size of bags in the constructed model is bounded. The number of child bags attached to each element is bounded by the number of existential restrictions in the TBox of $\Kk$, so the degree in the constructed model is also bounded. The complexity bound follows like in  Lemma~\ref{lem:algo-connected}.
\end{proof}

\end{document}